\documentclass[journal]{IEEEtran}

\usepackage{cite}
\usepackage{graphicx}
\usepackage{dcolumn}
\usepackage{bm}
\usepackage{amssymb}
\usepackage{amssymb}

\def\>{\rangle}

\newtheorem{theorem}{Theorem}

\newtheorem{remark}{Remark}
\newtheorem{example}{Example}

\begin{document}
%
\title{Nonlinear quantum input-output analysis using Volterra series}
%
%
\author{Jing~Zhang ~\IEEEmembership{Member, IEEE}, Yu-xi Liu, Re-Bing Wu ~\IEEEmembership{Member, IEEE}, Kurt Jacobs, Sahin Kaya Ozdemir, Lan Yang, Tzyh-Jong~Tarn ~\IEEEmembership{Life Fellow, IEEE}, Franco Nori
\thanks{Manuscript received ; first revised ; second revised ; third revised .
J. Zhang and R.~B. Wu are supported by the National Natural
Science Foundation of China under Grant Nos. 61174084, 61134008,
60904034, and Y.-X. Liu is supported by this foundation under
Grant Nos. 10975080, 61025022, 60836010. Y.-X. Liu and J. Zhang
are supported by the National Basic Research Program of China (973
Program) under Grant No. 2014CB921401, the Tsinghua University
Initiative Scientific Research Program, and the Tsinghua National
Laboratory for Information Science and Technology (TNList)
Cross-discipline Foundation. K. Jacobs is partially supported by
the US National Science Foundation projects PHY-1005571 and
PHY-1212413, and by the Army Research Office MURI project grant
W911NF-11-1-0268. L. Yang and S.~K. Sahin are supported by the ARO
grant No. W911NF-12-1-0026 and the NSFC under Grant No. 61328502.
F.N. is partially supported by the RIKEN iTHES Project, MURI
Center for Dynamic Magneto-Optics, and a Grant-in-Aid for
Scientific Research (S).}
\thanks{J. Zhang and R.-B. Wu are with the Department of Automation, Tsinghua University, Beijing
100084, P. R. China.(e-mail: jing-zhang@mail.tsinghua.edu.cn;
rbwu@tsinghua.edu.cn)}
\thanks{Y.-X. Liu is with the Institute of Microelectronics, Tsinghua University,
Beijing 100084, P. R. China.(e-mail:
yuxiliu@mail.tsinghua.edu.cn)}
\thanks{K. Jacobs is with the Department of Physics, University of Massachusetts at Boston, Boston, MA 02125, USA.(e-mail: kurt.jacobs@umb.edu)}
\thanks{S.~K. Ozdemir, L. Yang, and T.-J. Tarn are with the Electrical and Systems Engineering, Washington University,
St. Louis, Missouri 63130, USA.(e-mail:
ozdemir.sahin.kaya@gmail.com; yang@seas.wustl.edu;
tarn@wuauto.wustl.edu)}
\thanks{F. Nori is with the CEMS, RIKEN, Wako-shi, Saitama 351-0198, Japan and Physics Department, The University of Michigan, Ann Arbor, Michigan 48109-1040, USA.(e-mail:fnori@riken.jp)}
\thanks{J. Zhang, R.-B. Wu, Y.-X. Liu, and T.-J. Tarn are also with the Center for Quantum
Information Science and Technology, Tsinghua National Laboratory
for Information Science and Technology, Beijing 100084, P. R.
China.}}
%
%
%
\markboth{To be submitted to IEEE TRANSACTIONS ON AUTOMATIC
CONTROL}{Shell \MakeLowercase{\textit{et al.}}: Bare Demo of
IEEEtran.cls for Journals}
%



\maketitle

\begin{abstract}
Quantum input-output theory plays a very important role for
analyzing the dynamics of quantum systems, especially large-scale
quantum networks. As an extension of the input-output formalism of
Gardiner and Collet, we develop a new approach based on the
quantum version of the Volterra series which can be used to
analyze nonlinear quantum input-output dynamics. By this approach,
we can ignore the internal dynamics of the quantum input-output
system and represent the system dynamics by a series of kernel
functions. This approach has the great advantage of modelling
weak-nonlinear quantum networks. In our approach, the number of
parameters, represented by the kernel functions, used to describe
the input-output response of a weak-nonlinear quantum network,
increases linearly with the scale of the quantum network, not
exponentially as usual. Additionally, our approach can be used to
formulate the quantum network with both nonlinear and
nonconservative components, e.g., quantum amplifiers, which cannot
be modelled by the existing methods, such as the
Hudson-Parthasarathy model and the quantum transfer function
model. We apply our general method to several examples, including
Kerr cavities, optomechanical transducers, and a particular
coherent feedback system with a nonlinear component and a quantum
amplifier in the feedback loop. This approach provides a powerful
way to the modelling and control of nonlinear quantum networks.
\end{abstract}


%
\IEEEpeerreviewmaketitle

\section{Introduction}\label{s1}
%
%
%
%
\PARstart{T}here has been tremendous progress in the last few
years in the fields of quantum communication networks and quantum
internet~\cite{HJKimble,RvanMeterbook,JZhangSR}, quantum
biology~\cite{NLambert}, quantum chemistry, hybrid quantum
circuits~\cite{ZLXiang}, quantum computing and quantum
simulation~\cite{IGeorgescu,IBuluta}, and quantum
control~\cite{Alessandro,HMabuchiIJRNC:2005,PRouchon:2008,CBrif:2010,JZhangPhysRep,DYDong:2010,CAltafini:2012,GHuang:1983,HMWisemanbook,MRJames:2005,HMabuchiScience:2002,SHabibLAS:2002,Mirrahimi,NYamamoto,Altafini3,JZhang4,Wang,BQi,WCuiPRA:2013}.
These progresses pave the way to the development of large-scale
quantum networks. Although scalable quantum networks exhibit
advantages in information processing and transmission, many
problems are still left to be solved to model such complex quantum
systems. Different approaches have been proposed to analyze
quantum networks, among which the input-output formalism of
Gardiner and Collet~\cite{Gardiner,Gardiner2} is a useful tool to
describe the input-output dynamics of such systems. In fact, using
the input-output response to analyze and control the system
dynamics is a standard method in engineering. The quantum
input-output theory~\cite{Gardiner,Gardiner2} has been extended to
cascaded-connected quantum systems~\cite{Gardiner3,Carmichael} and
even more complex Markovian feedforward and feedback quantum
networks, including both dynamical and static
components~\cite{James,Gough,Squeezing_by_coherent_feedback,HMWiseman4,Lloyd,Mabuchi,RHamerly,Chia,SIida,JKerckhoff,Nurdin,Maalouf,JZhang,JZhang2,GFZhang}.
Two main different formulations are proposed in the literature to
model such quantum input-output networks: (i) the time-domain
Hudson-Parthasarathy formalism~\cite{Hudson}, which can be
considered as the extension of the input-output theory developed
by Gardiner and Collet; and (ii) the frequency-domain quantum
transfer function formalism~\cite{MYanagisawa1,AJShaiju}.

In the existing literature, quantum input-output theory is mainly
applied to optical systems, in which the ``memory'' effects of the
environment are negligibly small and the nonlinear effects of the
systems are weak and thus sometimes omitted. These lead to the
development of the Markovian and linear quantum input-output
network theory~\cite{James,MYanagisawa1,HMWisemanPRL2005}. In more
general cases, such as in mesoscopic solid-state systems, both the
linear and Markovian assumptions may not be valid. Recently, the
Markovian quantum input-output theory has been extended to the
non-Markovian case for single quantum input-output
components~\cite{LDiosi} or even quantum networks~\cite{JZhang3}.
However, how to model and analyze nonlinear quantum input-output
systems is still an open problem.

Recent experimental
progresses~\cite{Mabuchi,SIida,BLHigginsNature:2007,GYXiangNatPhoton:2011,HYonezawaScience:2012,CSayrinNature:2011,EGavarttinNatnano:2012},
especially those in solid-state quantum
circuits~\cite{JKerckhoff,RVijayNature2012,DRisteNature:2013,DRistePRL:2012,SShankerNature:2013},
motivate us to find some ways to analyze nonlinear quantum
input-output networks. It should be pointed out that the
Hudson-Parthasarathy model~\cite{Hudson} can in principle
formulate particular nonlinear quantum input-output systems.
However, it cannot be applied to more general cases, such as those
with both nonlinear components and quantum amplifiers. An
additional problem yet to be solved is the computational
complexity for modelling large-scale quantum input-output
networks. The Hudson-Parthasarathy model gives the input-output
response in terms of the internal system dynamics, and this will
lead to an exponential increase of the computational complexity
when we apply it to large-scale quantum networks composed of many
components. For most cases, this exponentially-increased
large-scale model contains redundant information. Not all the
internal degrees of freedom are necessary to be known. For
example, let us consider a quantum feedback control system
composed of the controlled system and the controller in the
feedback loop. We may not be interested in the internal dynamics
of the controller, but only concern how the controller modifies
the signal fed into it. This can be obtained by an input-output
response after averaging over the internal degrees of freedom of
the controller, which may greatly reduce the computational
complexity of the quantum network analysis. For linear quantum
networks, such an input-output response can be obtained by the
quantum transfer function model. However, it may not be applied to
nonlinear quantum networks.

To solve all these problems, we establish a nonlinear quantum
input-output formalism based on the so-called Volterra
series~\cite{DGorman:1965,LOChua1,LOChua2,MKielkiewicz:1970,JCPeyton:1991}.
This formalism gives a simpler form to model weak-nonlinear
quantum networks. This paper is organized as follows: a brief
review of quantum input-output theory is first presented in
Sec.~\ref{s2}, and then the general form of the Volterra series
for m-port quantum input-output systems is introduced in
Sec.~\ref{s3}. The Volterra series approach is extended in
sec.~\ref{s4} to more general cases in the frequency domain to
analyze more complex quantum input-output networks with multiple
components connected in series products and concatennation
products. Our general results are then applied to several examples
in Sec.~\ref{s5}. Conclusions and discussion of future work are
given in Sec.~\ref{s6}.

\section{Brief review of quantum input-output theory}\label{s2}

\subsection{Gardiner-Collet input-output formalism}\label{s21}

The original model of a general quantum input-output system is a
plant interacting with a bath. Under the Markovian approximation
(in which the coupling strengths between the system and different
modes of the bath are assumed to be constants for all
frequencies~\cite{Gardiner,Gardiner2}), an arbitrary system
operator $ Z\!\left( t \right) $ satisfies the following quantum
stochastic differential equation (QSDE)
\begin{eqnarray}\label{Quantum stochastic differential equation}
\dot{Z}&=&-i\left[Z,H_S\right] + \frac{1}{2} \left\{ {\bf
L}^{\dagger} \left[ Z,
{\bf L} \right] + \left[ {\bf L}^{\dagger}, Z \right] {\bf L} \right\} \nonumber \\
&& + \left\{ {\bf b}_{\rm in} \left[ {\bf L}^{\dagger}, Z \right]
+ \left[ Z, {\bf L} \right] {\bf b}_{\rm in}^{\dagger} \right\},
\end{eqnarray}
with ${\bf L}=\left(L_1,\cdots,L_m\right)^T$, where $L_i$'s are
the system operators representing the dissipation channels of the
system coupled to the input fields.
Let ${\bf b}_{\rm in}\!\left(t\right)=\left[b_{1,{\rm
in}}\!\left(t\right),\cdots,b_{m,{\rm
in}}\!\left(t\right)\right]^T$ and output field ${\bf b}_{\rm
out}\!\left(t\right)=\left[b_{1,{\rm
out}}\!\left(t\right),\cdots,b_{m,{\rm
out}}\!\left(t\right)\right]^T$ be the time-varying input fields
that are fed into the system and the output fields being about to
propagate away, one has the relation
\begin{equation}\label{Output equation under the Markovian approximation}
{\bf b}_{\rm out}\!\left( t \right) = {\bf b}_{\rm in}(t) + {\bf
L}(t).
\end{equation}
This is the standard Gardiner-Collet input-output relation.

\subsection{Hudson-Parthasarathy model}\label{s22}

The Gardiner-Collet input-output theory can be extended to more
general case to include the static components such as quantum beam
splitters by the Hudson-Parthasaraty model~\cite{Gough}. A more
general multi-input and multi-output open quantum systems can be
characterized by the following tuple of parameters
\begin{equation}\label{SLH}
G = \left( {\bf S}, {\bf L}, H \right),
\end{equation}
where $H$ is the internal Hamiltonian of the system; $ {\bf S } $
is a $ n \times n $ unitary scattering matrix induced by the
static components. The notations given in Eq.~(\ref{SLH}) can be
used to describe a wide range of dynamical and static systems. For
example, the traditional quantum input-output systems represented
by Eqs.~(\ref{Quantum stochastic differential equation}) and
(\ref{Output equation under the Markovian approximation}) can be
written as $ G_{LH} = \left( I, {\bf L}, H \right)$, and the
quantum beam splitter with scattering matrix ${\bf S}$ can be
represented by $ G_{BS} = \left( {\bf S}, 0, 0 \right) $.

To obtain the dynamics of the input-output system given by
Eq.~(\ref{SLH}), we first introduce the quantum Wiener process $
{\bf B} \left( t \right) $ and the quantum Poisson process $ {\bf
\Lambda } \left( t \right) $ as
\begin{equation}\label{Quantum Wiener and Poisson processes}
{ \bf B } \left( t \right) = \left(%
\begin{array}{c}
  B_1 \\
  \vdots \\
  B_n \\
\end{array}%
\right), \quad { \bf \Lambda } \left( t \right) = \left(%
\begin{array}{ccc}
  B_{11} & \cdots & B_{1n} \\
  \vdots & \ddots & \vdots \\
  B_{n1} & \cdots & B_{nn} \\
\end{array}%
\right),
\end{equation}
which are defined by
\begin{equation}\label{Definition of quantum Wiener and Poisson processes}
B_i \left( t \right) = \int_0^t b_{i,{\rm in}} \left( \tau \right)
d \tau, \quad B_{ij} \left( t \right) = \int_0^t b_{i,{\rm
in}}^{\dagger} \left( \tau \right) b_{j,{\rm in}} \left( \tau
\right) d \tau.
\end{equation}
In the Heisenberg picture, the system operator $ Z \left( t
\right)$ satisfies the following quantum stochastic differential
equation
\begin{eqnarray}\label{Quantum stochastic differential equation of SLH}
d Z  & = & \left\{ \mathcal{L}_{ {\bf L}} \left( Z  \right) - i
\left[ Z, H \right] \right\} dt + d {\bf B}^{\dagger} {\bf
S}^{\dagger} \left[ Z , {\bf L} \right] \nonumber \\
& & + \left[ {\bf L}^{\dagger} , Z \right] {\bf S} d {\bf B} +
{\bf tr} \left\{ \left[ {\bf S}^{\dagger} Z {\bf S} - Z \right] d
{\bf \Lambda}^T \right\},
\end{eqnarray}
where the Liouville superoperator $ \mathcal{L}_{ {\bf L }
}\left(\cdot \right)$ is defined by
\begin{eqnarray}\label{Dissipation Liouville superoperator}
\mathcal{L}_{ {\bf L } } \left( X \right) &=& \frac{1}{2} {\bf
L}^{\dagger} \left[ X, {\bf L } \right] + \frac{1}{2} \left[ {\bf
L}^{\dagger}, X \right] {\bf L}\nonumber\\
&&=\sum_{j=1}^n \left\{ \frac{1}{2} L_j^{\dagger} \left[ X, L_j
\right] + \frac{1}{2} \left[ L_j^{\dagger}, X \right] L_j
\right\},
\end{eqnarray}
which is of the standard Lindblad form. Similar to
Eq.~(\ref{Output equation under the Markovian approximation}), we
can obtain the following input-output relation
\begin{eqnarray}\label{Input-output relation of quantum Wiener and Poisson processes}
d {\bf B}_{\rm out} & = & {\bf S} d {\bf B} + {\bf L} dt,
\nonumber \\
d {\bf \Lambda}_{\rm out} & = & {\bf S}^* d {\bf \Lambda} {\bf
S}^T + {\bf S}^* d {\bf B}^* {\bf L}^T + {\bf L}^* d {\bf B}^T
{\bf
S}^T + {\bf L}^* {\bf L}^T dt,\nonumber \\
\end{eqnarray}
where $d{\bf B}_{\rm out}$ and $d{\bf \Lambda}_{\rm out}$ are the
output fields corresponding to the quantum Wiener process $ d{\bf
B} $ and Poisson process $ d\Lambda $.

\subsection{Quantum transfer function model}\label{s23}

The Gardiner-Collet input-output theory, or the more general
Hudson-Parthasarathy model introduced in subsections~\ref{s21} and
\ref{s22}, can be used to represent a large class of quantum
input-output systems. However, the Hudson-Parthasarathy model is
complex if the interior degrees of freedom of the system are very
high. The quantum transfer function model can be applied to some
cases that the Hudson-Parthasarathy model is invalid or
inefficient.

Different from the Hudson-Parthasarathy model, which is in the
time domain, the quantum transfer function model is a
frequency-domain approach~\cite{MYanagisawa1,AJShaiju} and can
only be applied to linear quantum input-output systems. The system
we consider is composed of $r$ harmonic oscillators $ \left\{ a_j:
j=1,\cdots,r \right\} $, which satisfy the following canonical
commutation relations
\begin{eqnarray*}
\left[ a_j, a_k^{\dagger} \right] = \delta_{jk}, \quad \left[ a_j,
a_k \right] = \left[ a_j^{\dagger}, a_k^{\dagger} \right] =0.
\end{eqnarray*}
We are interested in a general linear quantum system, which, in
the $\left( {\bf S}, {\bf L}, H\right)$ notation given by
Eq.~(\ref{SLH}), satisfies the following conditions: (i) the
dissipation operators $L_j$'s are linear combinations of $a_k$,
i.e., $L_j$ can be written as $L_j=\sum_k c_{jk} a_k$; and (ii)
the system Hamiltonian $H$ is a quadratic function of $a_k$, i.e.,
$H=\sum_{jk}\omega_{jk}a_j^{\dagger}a_k$. Under these conditions,
we can obtain the following equivalent expression of
Eq.~(\ref{SLH}):
\begin{equation}\label{SCOmega}
G = \left( {\bf S}, C, \Omega \right),
\end{equation}
where
\begin{eqnarray*}
C=\left(%
\begin{array}{ccc}
  c_{11} & \cdots & c_{1r} \\
  \vdots & \ddots & \vdots \\
  c_{r1} & \cdots & c_{rr} \\
\end{array}%
\right),\quad \Omega=\left(%
\begin{array}{ccc}
  \omega_{11} & \cdots & \omega_{1r} \\
  \vdots & \ddots & \vdots \\
  \omega_{r1} & \cdots & \omega_{rr} \\
\end{array}%
\right).
\end{eqnarray*}

Let us introduce an operator vector called the state vector of the
system ${\bf a}= \left( a_1, \cdots, a_r \right)^T$, then from
Eqs.~(\ref{Quantum stochastic differential equation of SLH}) and
(\ref{Input-output relation of quantum Wiener and Poisson
processes}), we can obtain the following Heisenberg-Langevin
equation and input-output relation
\begin{eqnarray}
\dot{{\bf a}}\!\left( t \right) & = & A\,{\bf a}\!\left( t \right)
- C^{\dagger} {\bf S} {\bf b}_{\rm in}\! \left( t \right),
\label{Heisenberg-Langevin equation} \\
{\bf b}_{\rm out} & = & {\bf S} {\bf b}_{\rm in}\! \left( t
\right) + C\,{\bf a}\!\left( t \right), \label{Input-output
relation}
\end{eqnarray}
where $A=-C^{\dagger}C/2-i\Omega$. Such kind of linear equations
can be solved in the frequency domain. To show this, let us
introduce the Laplace transform which is defined for $ {\rm
Re}\left( s \right)
> 0 $ by
\begin{equation}\label{Laplace transform of the linear equation}
R \left( s \right) = \int_0^{\infty} \exp{ \left( - s t \right) }
R \left( t \right) d t.
\end{equation}
In the frequency domain, Eqs.~(\ref{Heisenberg-Langevin equation})
and (\ref{Input-output relation}) can be solved as
\begin{eqnarray}
& {\bf a}\! \left( s \right) = - \left( s I_r - A \right)^{-1}
C^{\dagger} {\bf S} {\bf b}_{\rm in}\!\left( s \right), &
\label{Heisenberg-Langevin equation in the
frequency domain} \\
& {\bf b}_{\rm out} \left( s \right) = {\bf S} {\bf b}_{\rm in}
\left( s \right) + C {\bf a} \left( s \right). &
\label{Input-output relation in the frequency domain}
\end{eqnarray}
Then, we can obtain the input-output relation of the whole system
\begin{equation}\label{Input-output relation by quantum transfer function model}
{\bf b}_{\rm out} \left( s \right) = \Xi \left( s \right) {\bf
b}_{\rm in} \left( s \right),
\end{equation}
where $ \Xi \left( s \right) $ is the transfer function of the
linear quantum system, which can be calculated by
\begin{equation}\label{Quantum transfer function}
\Xi \left( s \right) = {\bf S} - C \left( s I_r- A \right)^{ -1 }
C^{\dagger} {\bf S}.
\end{equation}
The input-output relation~(\ref{Input-output relation by quantum
transfer function model}) shows the linear map between the input
and output of the linear quantum system given by
Eqs.~(\ref{Heisenberg-Langevin equation}) and (\ref{Input-output
relation}).

\section{Nonlinear response by Volterra series}\label{s3}

The Gardiner and Collet's quantum input-output
model~\cite{Gardiner,Gardiner2}, or more generally the
Hudson-Parthasarathy model, give a general form of the quantum
input-output response, but there are internal degrees of freedom
determined by Eq.~(\ref{Quantum stochastic differential
equation}). Sometimes, it is not easy to use these models to
describe an input-output system, especially for nonlinear systems
of which the interior degrees of freedom are extremely high or
even infinite-dimensional. However, the complexity of the quantum
input-output model can be greatly reduced if we average out the
interior dynamics. For linear quantum systems, such a reduction
process leads to the quantum transfer function model, in which the
quantum input-output response is represented in the frequency
domain as a linear input-output relation with a proportional gain
called {\it the quantum transfer function}. As an extension of
this method, we will show that the Volterra series can be used to
describe the input-output response for more general nonlinear
quantum input-output systems.
\\[0.002cm]

\begin{theorem}(Nonlinear quantum input-output response by Volterra
series)\label{Nonlinear quantum input-output response by Volterra
series}

The quantum input-output relation of a general m-port quantum
system with input field ${\bf b}_{\rm
in}\!\left(t\right)=\left[b_{1,{\rm
in}}\!\left(t\right),\cdots,b_{m,{\rm
in}}\!\left(t\right)\right]^T$ and output field ${\bf b}_{\rm
out}\!\left(t\right)=\left[b_{1,{\rm
out}}\!\left(t\right),\cdots,b_{m,{\rm
out}}\!\left(t\right)\right]^T$ can be expressed as the following
{\it Volterra series}

\begin{eqnarray}\label{Volterra series of multi-input quantum input-output system}
b_{j, {\rm out}}^{\pm}\!\left(t\right)&=&\sum_{n=1}^{+\infty}\!
\int_0^t\!\!\int_0^{\tau_1}\!\!\!\!\cdots\!\!\!\int_0^{\tau_{n-1}}\!\!\!\!\!\!\!\!b_{
j_1, {\rm in}}^{\left(i_1\right)}\!\left(\tau_1\right)\cdots
b_{j_n, {\rm in}}^{\left(i_n\right)}\!\left(\tau_n\right)\nonumber\\
&& k^{\pm,j_1\cdots j_n}_{j,i_1\cdots
i_n}\!\left(t-\tau_1,\cdots,\tau_{n-1}-\tau_n\right) d {\bf \tau},
\end{eqnarray}
where $i_k=\pm$ and $j_k=1,\cdots,m$; $b_{j,{\rm
in}}^{(-)}=b_{j,{\rm in}}$ and $b_{j,{\rm out}}^{(-)}=b_{j,{\rm
out}}$; $b_{j,{\rm in}}^{(+)}$ and $b_{j,{\rm out}}^{(+)}$ are the
conjugate operators of $b_{j,{\rm in}}$ and $b_{j,{\rm out}}$; and
$\left\{k^{i,j_1\cdots j_n}_{j,i_1\cdots
i_n}\!\left(\tau_1,\cdots,\tau_n\right)\right\}$ are the kernel
functions of the $m$-port input-output system. Here, we omit the
sum of the indices $i_k=\pm$ and $j_k=1,\cdots,m$ by the Einstein
summation convention.
\begin{figure}[t]
\centerline{\includegraphics[width=6.7 cm]{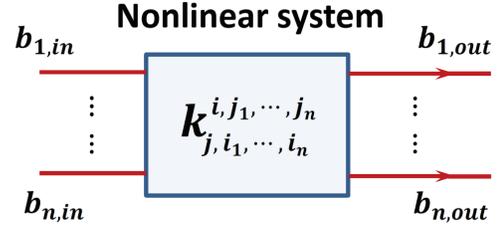}}
\caption{(color online) An $n$-port nonlinear input-out component
with kernel functions $\left\{k^{i,j_1\cdots j_n}_{j,i_1\cdots
i_n}\!\left(\tau_1,\cdots,\tau_n\right)\right\}$.}\label{Fig of
n-port quantum input-output component}
\end{figure}

\end{theorem}

\begin{proof}
To prove the theorem, let us first assume that the dynamical Lie
algebra of the dynamical system given by Eq.~(\ref{Quantum
stochastic differential equation}) is $\mathcal{L}={\rm
span}\left(\left\{-iX_{\alpha}\right\}\right)$, where
$-iX_{\alpha}$'s are the basis elements of the Lie algebra
$\mathcal{L}$ which satisfy the following commutation relation
\begin{equation}\label{Commutation relation of the dynamical Lie algebra}
[X_{\alpha},X_{\beta}]=-i\sum_{\gamma} C_{\alpha\beta}^{\gamma}
X_{\gamma}.
\end{equation}
$C_{\alpha\beta}^{\gamma}$'s are the structure constants of the
Lie algebra $\mathcal{L}$. Let us define an operator vector ${\bf
X}=\left(X_{\alpha}\right)$, of which the entries come from the
basis elements of the Lie algebra $\mathcal{L}$. From
Eq.~(\ref{Quantum stochastic differential equation}), we can
obtain the following dynamical equation for the operator vector
${\bf X}$
\begin{equation}\label{Bilinear equation for general quantum input-output components}
\dot{{\bf X}}\!\left(t\right)=A{\bf X}\left(t\right)+\left(B^*{\bf
b}_{\rm in}\!\left(t\right)+B {\bf b}^{\dagger}_{\rm
in}\right){\bf X}\!\left(t\right).
\end{equation}

From Eq.~(\ref{Bilinear equation for general quantum input-output
components}), we have the following formal series solution for the
above equation
\begin{eqnarray}\label{Iterative solution of bilinear
equation} {\bf X}\!\left(t\right)&=&e^{At}{\bf X}+\int_0^t
e^{A\left(t-\tau_1\right)}\left[ B^* {\bf b}_{\rm
in}\!\left(\tau_1\right)\right. \nonumber\\
&&\left.+B {\bf b}_{\rm in}^{\dagger}\!\left(\tau_1\right)\right]
{\bf X}\!\left(\tau_1\right) d \tau_1,
\end{eqnarray}
where ${\bf X}$ is the operator vector in the Schr\"{o}dinger
picture. By solving ${\bf X}\!\left(\tau_1\right)$ in the integral
of Eq.~(\ref{Iterative solution of bilinear equation}), we can
obtain the iterative solution
\begin{eqnarray*}
{\bf X}\!\left(t\right)&=&e^{At}{\bf X}+\int_0^t
e^{A\left(t-\tau_1\right)}\left[ B^* {\bf b}_{\rm
in}\!\left(\tau_1\right)\right.\\
&&\left.+B {\bf b}_{\rm in}^{\dagger}\!\left(\tau_1\right)\right]
e^{A \tau_1} {\bf
X} d \tau_1 \\
& & + \int_0^t e^{A\left(t-\tau_1\right)}\left[ B^* {\bf b}_{\rm
in}\!\left(\tau_1\right)+B {\bf b}_{\rm
in}^{\dagger}\!\left(\tau_1\right)\right] d\tau_1 \\
& & \int_0^{\tau_1} e^{A\left(\tau_1-\tau_2\right)} \left[ B^*
{\bf b}_{\rm in}\!\left(\tau_2\right)+B {\bf b}_{\rm
in}^{\dagger}\!\left(\tau_2\right)\right] {\bf
X}\!\left(\tau_2\right) d\tau_2.
\end{eqnarray*}
Solving the above equation in the same way, we can obtain the
following series solution of Eq.~(\ref{Bilinear equation for
general quantum input-output components})
\begin{eqnarray}\label{Iterative expansion of the operator vector}
{\bf X}\!\left(t\right)&=&e^{At}{\bf X}+\sum_{n=1}^{\infty}\!
\int_0^t\!\!\!\!\cdots\!\!\!\int_0^{\tau_{n-1}}\!\!\!\!\!\!\!\!
e^{A\left(t-\tau_1\right)}\left[ B^* {\bf b}_{\rm
in}\left(\tau_1\right)\right.\nonumber\\
&&\left.+B {\bf b}_{\rm in}^{\dagger}\!\left(\tau_1\right)\right]
e^{A \left(\tau_1-\tau_2\right)}\cdots e^{-A\tau_n} {\bf X} d
\tau_1\cdots d\tau_n.\nonumber\\
\end{eqnarray}

Since $\left\{-iX_{\alpha}\right\}$ is the basis of the dynamical
Lie algebra of the quantum input-output system, the system
operator ${\bf L}$ in the output equation~(\ref{Quantum stochastic
differential equation}) can be written as the linear combination
of $\left\{X_{\alpha}\right\}$, i.e.,
\begin{equation}\label{Linear combination of dissipation operator}
{\bf L}=\sum_{\alpha} l_{\alpha} X_{\alpha},
\end{equation}
where $l_{\alpha}\in\mathbb{C}^n$. Let us then assume that the
total input-output system composed of the internal degrees of
freedom and the external input field is initially in a separable
state $\rho_{\rm tot}\!\left(0\right)=\rho_0\otimes\rho_b$, where
$\rho_0$ and $\rho_b$ are, respectively, the initial states of the
internal system and the external input field. If we average over
the internal degrees of freedom of the quantum input-output
system, the output equation~(\ref{Output equation under the
Markovian approximation}) can be rewritten as
\begin{equation}\label{Output equation under the Markovian approximation in average}
{\bf b}_{\rm out}\!\left( t \right) = {\bf b}_{\rm in}\!(t) +
\sum_{\alpha} {\bf l}_{\alpha} \langle X_{\alpha}\!(t) \rangle_0,
\end{equation}
where $\langle X_{\alpha}\!\left(t\right) \rangle_0={\rm
tr}\left[X_{\alpha}\!\left(t\right)\rho_0\right]$. By substituting
Eq.~(\ref{Iterative expansion of the operator vector}) into
Eq.~(\ref{Output equation under the Markovian approximation in
average}), we can obtain the Volterra series of a general
nonlinear quantum input-output component given by
Eq.~(\ref{Volterra series of multi-input quantum input-output
system}).
\end{proof}

As shown in Eq.~(\ref{Volterra series of multi-input quantum
input-output system}), the system input-output response is fully
determined by the set of parameters $\left\{k^{\pm,j_1\cdots
j_n}_{j,i_1\cdots i_n}\!\left(\tau_1,\cdots,\tau_n\right)\right\}$
called {\it Volterra kernels}. It is also shown in the proof of
theorem~\ref{Nonlinear quantum input-output response by Volterra
series} that these kernel functions are just determined by the
high-order quantum correlations of the interior dynamics of the
quantum input-output system. Notice that the quantum Volterra
series is different from the classical Volterra series because the
terms like ${\bf b}_{\rm
in}^{\left(i_1\right)}\left(\tau_1\right),\cdots,{\bf b}_{\rm
in}^{\left(i_n\right)}\left(\tau_n\right)$ do not commute with
each other and the vacuum fluctuations in the input field should
be considered when we analyze quantum input-output response.

Different from the Volterra series used for for classical systems,
the kernel functions of the quantum Volterra series method we
present here should satisfy additional physically-realizable
conditions~\cite{James,AJShaiju} constrained by the theory of
quantum mechanics. For example, for a Markovian input-output
system~\cite{Gough}, the commutation relation should be preserved
from the input field to the output field, i.e., $[b_{\rm
out}\!\left(t\right),b_{\rm
out}^{\dagger}\!\left(t^{\prime}\right)]=[b_{\rm
in}\!\left(t\right),b_{\rm
in}^{\dagger}\!\left(t^{\prime}\right)]$. This leads to additional
equality constraints for the kernel functions $k^{\pm,j_1\cdots
j_n}_{j,i_1\cdots i_n}\!\left(\tau_1,\cdots,\tau_n\right)$.

Although the right side of Eq.~(\ref{Volterra series of
multi-input quantum input-output system}) is an infinite series,
i.e., a series with infinitely many terms, we can use its finite
truncations to represent the input-output response under
particular conditions. In fact, for linear systems, there are only
linear terms in Eq.~(\ref{Volterra series of multi-input quantum
input-output system}). Motivated by this consideration, we then
study a weak-nonlinear quantum system with $r$ internal modes
given by the annihilation (creation) operators $a_{i=1,\cdots,r}$
($a^{\dagger}_{i=1,\cdots,r}$). For such a weak nonlinear quantum
system, the system Hamiltonian $H$ and the dissipation operator
${\bf L}$ in Eq.~(\ref{Quantum stochastic differential equation})
can be expressed as
\begin{equation}\label{Hamiltonian and dissipation operator of weak-nonlinear quantum system}
H=H_l+\mu H_{nl},\quad{\bf L}={\bf L}_l+\mu{\bf L}_{nl},
\end{equation}
where $H_l$ and ${\bf L}_l$ are quadratic and linear functions of
the annihilation and creation operators $a_i$ and $a_i^{\dagger}$;
$H_{nl}$ and ${\bf L}_{nl}$ are higher-order nonlinear terms of
$a_i$ and $a^{\dagger}_i$; and $\mu$ is a parameter introduced to
determine the nonlinear degree of the system. For a weak-nonlinear
system, we have $\mu\ll 1$. The following theorem shows that we
can use the finite truncation up to low-order terms and omit
higher-order nonlinear terms in the Volterra series for
weak-nonlinear systems satisfying Eq.~(\ref{Hamiltonian and
dissipation operator of weak-nonlinear quantum system}).
\\[0.002cm]

\begin{theorem}(Volterra series for weak-nonlinear systems)\label{Volterra series for weak-nonlinear system}

For a weak-nonlinear quantum input-output system with Hamiltonian
and dissipation operators given by Eq.~(\ref{Hamiltonian and
dissipation operator of weak-nonlinear quantum system}), the
Volterra series for this quantum input-output system can be
written as
\begin{eqnarray}\label{Volterra series of weak-nonlinear quantum input-output system}
b_{j, {\rm out}}^{\pm}\!\left(t\right)&=&\int_0^t k^{\pm\,j_1}_{j,i_1}(t-\tau)b^{(i)}_{j_1,{\rm in}}\left(\tau_1\right)d\tau_1\nonumber\\
&&+\mu\int_0^t\!\!\int_0^{\tau_1}d\tau_1 d\tau_2\,b_{ j_1, {\rm
in}}^{\left(i_1\right)}\!\left(\tau_1\right) b_{j_2, {\rm in}}^{\left(i_2\right)}\!\left(\tau_2\right)\nonumber\\
&&
k^{\pm,j_1j_2}_{j,i_1i_2}\!\left(t-\tau_1,\tau_1-\tau_2\right)+o\left(\mu\right).
\end{eqnarray}
\\[0.002cm]
\end{theorem}

The proof of the theorem is given in the appendix. It can be
easily seen that Eq.~(\ref{Volterra series of weak-nonlinear
quantum input-output system}) is just the traditional convolution
representation, or equivalently the transfer function
representation, when $\mu=0$, which corresponds to linear quantum
systems.

\section{Frequency analysis of nonlinear quantum input-output networks}\label{s4}

The Volterra series can be expressed as a simpler form in the
frequency domain, especially for quantum networks with several
components. In the frequency domain, the quantum input-output
relation can be rewritten as
\begin{eqnarray}\label{Volterra series of general quantum input-output system in the frequency domain}
b_{\rm
j,out}^{\pm}\!\left(\omega\right)&=&\sum_{n=1}^{+\infty}\mathop{
\int\cdots\int}_{\omega_1+\cdots \omega_n=\omega} b_{j_1,{\rm
in}}^{\left(i_1\right)}\!\left(\omega_1\right)\cdots b_{j_n,{\rm in}}^{\left(i_n\right)}\!\left(\omega_n\right)\nonumber\\
&& \chi^{\pm,j_1\cdots j_n}_{j,i_1\cdots
i_n}\left(\omega_1,\cdots,\omega_n\right) d\omega_1\cdots\omega_n,
\end{eqnarray}
where
\begin{eqnarray*}
\chi_{j,i_1\cdots i_n}^{\pm,j_1\cdots
j_n}\!\left(\omega_1,\cdots,\omega_n\right)=K^{\pm_1\cdots
j_n}_{j,i_1\cdots
i_n}\left(\omega_1+\cdots\omega_n,\cdots,\omega_n\right)
\end{eqnarray*}
and $K^{\pm,j_1\cdots j_n}_{j,i_1\cdots
i_n}\left(\omega_1,\cdots,\omega_n\right)$ is the $n$-th order
Fourier transform of the kernel function $k^{\pm,j_1\cdots
j_n}_{j,i_1\cdots i_n}\!\left(\tau_1,\cdots,\tau_n\right)$ defined
by
\begin{eqnarray}\label{Generalized Fourier transform}
K^{\pm,j_1\cdots j_n}_{j,i_1\cdots
i_n}\!\left(\omega_1,\cdots,\omega_n\right)&=&\int_{-\infty}^{+\infty}e^{-i\omega_1t_1\cdots-i\omega_nt_n}\nonumber\\
&&k^{\pm,j_1\cdots j_n}_{j,i_1\cdots
i_n}\!\left(t_1,\cdots,t_n\right)dt_1\cdots dt_n.\nonumber\\
\end{eqnarray}
The coefficients $\left\{\chi_{j,i_1\cdots i_n}^{\pm,j_1\cdots
j_n}\right\}$ can be seen as the quantum version of the $n$-th
order nonlinear susceptibility coefficients.

The main merit of the Volterra series approach is that {\it it can
greatly reduce the computational complexity of quantum
input-output network analysis in the frequency domain}. In fact,
from Eq.~(\ref{Volterra series of general quantum input-output
system in the frequency domain}), we can find that the
input-output response of a multi-input nonlinear component is
fully determined by the quantum susceptibility coefficients
$\left\{\chi_{j,i_1\cdots i_n}^{i,j_1\cdots
j_n}\left(\omega_1,\cdots,\omega_n\right)\right\}$. The following
theorem shows that the quantum susceptibility coefficients of a
large-scale quantum network with several components can be
expressed as the polynomial functions of the lower-order quantum
susceptibility coefficients of each component.

\begin{theorem}(Susceptibility coefficients for networks)\label{Susceptibility coefficients for quantum networks}

The $n$-th order quantum susceptibility coefficients
$\left\{\chi_{j,i_1\cdots i_n}^{i,j_1\cdots
j_n}\left(\omega_1,\cdots,\omega_n\right)\right\}$ of a
multi-component quantum network can be expressed as polynomials of
lower-order quantum susceptibility coefficients of each component.

\end{theorem}

\begin{proof}
To prove our main results, we can see that an arbitrary nonlinear
quantum network can be decomposed into two basic types of
connections between different components, i.e., the concatenation
product and the series product~\cite{Gough}. Thus, we only need to
verify the main results for these two types of basic quantum
networks.
\begin{figure}[t]
\centerline{\includegraphics[width=8.8 cm]{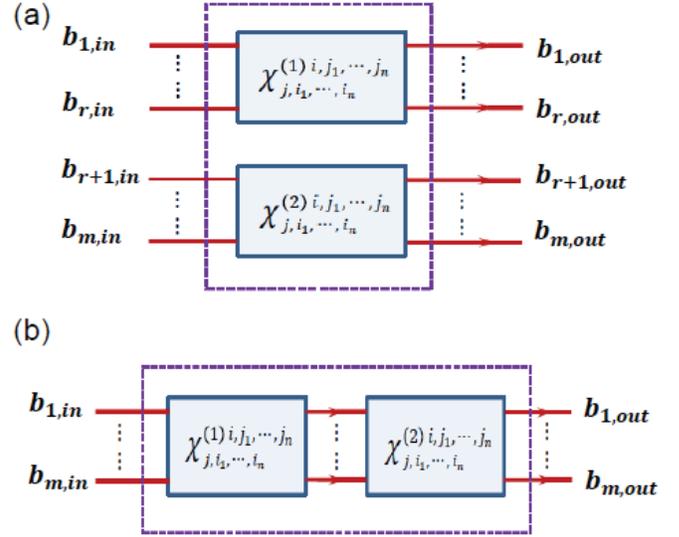}}
\caption{(color online) Two types of basic connections between
different components: (a) the concatenation product, in which two
components are simply assembled together without any connection
between them; and (b) the series product in which two components
are cascade-connected, i.e., the output of the first system is
taken as the input of the second system.
$\left\{\chi^{\left(1\right)\,\,i,j_1\cdots j_n}_{j,i_1\cdots
i_n}\!\left(\tau_1,\cdots,\tau_n\right)\right\}$ and
$\left\{\chi^{\left(2\right)\,\,i,j_1\cdots j_n}_{j,i_1\cdots
i_n}\!\left(\tau_1,\cdots,\tau_n\right)\right\}$ are the quantum
susceptibility coefficients of the two components.}\label{Fig of
the concatenation product and the series product}
\end{figure}

The concatenation product describes two components that are simply
assembled together without any connection between them [see
Fig.~\ref{Fig of the concatenation product and the series
product}(a)]. Let us assume that
$\left\{\chi^{\left(1\right)\,\,i,j_1\cdots j_n}_{j,i_1\cdots
i_n}\!\left(\tau_1,\cdots,\tau_n\right)\right\}$ and
$\left\{\chi^{\left(2\right)\,\,i,j_1\cdots j_n}_{j,i_1\cdots
i_n}\!\left(\tau_1,\cdots,\tau_n\right)\right\}$ are the quantum
susceptibility coefficients of the two components and
$\left\{\chi^{i,j_1\cdots j_n}_{j,i_1\cdots
i_n}\!\left(\tau_1,\cdots,\tau_n\right)\right\}$ are the quantum
susceptibility coefficients of the total system, then it can be
easily verified that
\begin{eqnarray}\label{Susceptibility coefficients of concatenation product}
\lefteqn{
\chi_{j,i_1\cdots i_n}^{i,j_1\cdots j_n}\left(\omega_1,\cdots,\omega_n\right)=}\nonumber\\
&&\left\{%
\begin{array}{ll}
    \chi_{j,i_1\cdots i_n}^{(1)\,\,i,j_1\cdots j_n}\left(\omega_1,\cdots,\omega_n\right), & i_k,j_k\in\left\{1,\cdots,r\right\}; \\
    \chi_{j,i_1\cdots i_n}^{(2)\,\,i,j_1\cdots j_n}\left(\omega_1,\cdots,\omega_n\right), & i_k,j_k\in\left\{r+1,\cdots,m\right\}; \\
    0, & \hbox{{\rm otherwise}.} \\
\end{array}%
\right.\nonumber\\
\end{eqnarray}

The series product can be used to describe two cascade-connected
components [see Fig.~\ref{Fig of the concatenation product and the
series product}(b)], i.e., the output of the first system is taken
as the input of the second system. The $n$-th quantum
susceptibility coefficients $\chi_{j,i_1\cdots i_n}^{i,j_1\cdots
j_n}$ of the quantum network in the series product can be
calculated by the following equation
\begin{eqnarray}\label{Susceptibility coefficients of series product}
\lefteqn{
\chi_{j,i_1\cdots i_n}^{i,j_1\cdots j_n}\left(\omega_1,\cdots,\omega_n\right)=}\nonumber\\
&&\sum_{\tiny\begin{array}{c}
  r=1,\cdots,n \\
  \alpha_1+\cdots+\alpha_r=n \\
\end{array}}
\chi_{j,k_1\cdots k_r}^{\left(2\right)\,\,i,l_1\cdots l_r}\left(\omega_1+\cdots+\omega_{\alpha_1},\cdots\right)\nonumber\\
&&\,\,\,\,\,\,\chi_{l_1,i_1\cdots
i_{\alpha_1}}^{\left(1\right)\,\,k_1,j_1\cdots
j_{\alpha_1}}\left(\omega_1,\cdots,\omega_{\alpha_1}\right)\cdots\nonumber\\
&&\,\,\,\,\,\,\chi_{l_r,i_{n-\alpha_r+1}\cdots
i_n}^{\left(1\right)\,\,k_r,j_{n-\alpha_r+1}\cdots
j_n}\left(\omega_{n-\alpha_r+1},\cdots,\omega_n\right),
\end{eqnarray}
where $\chi_{j,i_1\cdots i_n}^{\left(1\right)\,\,i,j_1\cdots j_n}$
and $\chi_{j,i_1\cdots i_n}^{\left(2\right)\,\,i,j_1\cdots j_n}$
are the quantum susceptibility coefficients of the two components.

From Eq.~(\ref{Susceptibility coefficients of concatenation
product}) and Eq.~(\ref{Susceptibility coefficients of series
product}), we can see that the $n$-th quantum susceptibility
coefficients of a quantum network in the concatenation product and
series product can be expressed as the polynomials of lower-order
quantum susceptibility coefficients of the components in the
quantum network. That completes the proof of the theorem.
\end{proof}

\begin{remark}
Theorem~\ref{Susceptibility coefficients for quantum networks}
shows that the computational complexity to describe the
input-output response of a multi-component nonlinear quantum
network increases linearly with the number of the components in
the quantum network, in comparison to the traditional
exponentially increasing complexity for describing a complex
quantum input-output network.
\end{remark}

\begin{figure}[t]
\centerline{\includegraphics[width=8.8 cm]{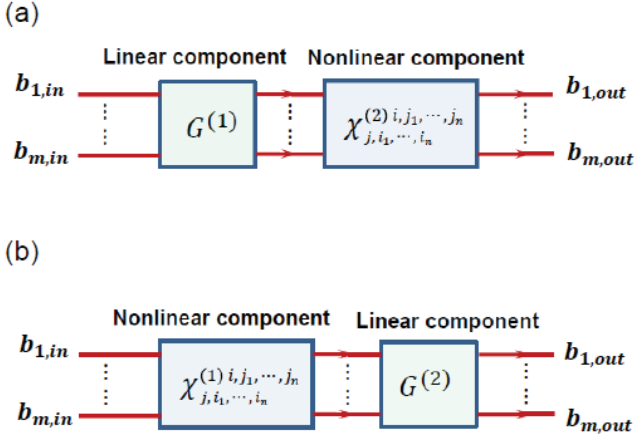}}
\caption{(color online) Series product networks with a linear
component and a nonlinear component: (a) the input field is first
fed into the linear component with quantum transfer function
$G^{\left(1\right)}\!\left(i\omega\right)$ and then transmits
through a nonlinear component with quantum susceptibility
coefficients $\chi_{j,i_1\cdots
i_n}^{\left(2\right)\,\,i,j_1\cdots
j_n}\left(\omega_1,\cdots,\omega_n\right)$; (b) the two components
are connected in the opposite way: the input field is first fed
into a nonlinear component with quantum susceptibility
coefficients $\chi_{j,i_1\cdots
i_n}^{\left(1\right)\,\,i,j_1\cdots
j_n}\left(\omega_1,\cdots,\omega_n\right)$ and then a linear
component with quantum transfer function
$G^{\left(2\right)}\!\left(i\omega\right)$.}\label{Fig of the
series product networks with a linear component and a nonlinear
component}
\end{figure}
As examples of multi-component nonlinear quantum networks, let us
consider a quantum network in which a linear component is
cascade-connected to a nonlinear component. This can be divided
into two different cases (see Fig.~\ref{Fig of the series product
networks with a linear component and a nonlinear component}):

(i) The input field is first fed into a linear component with the
quantum transfer function~\cite{MYanagisawa1}
$G^{\left(1\right)}\!\left(s\right),\,s\in\mathbb{C}$, and then
transmits through a nonlinear component with quantum
susceptibility coefficients $\chi_{j,i_1\cdots
i_n}^{\left(2\right)\,\,i,j_1\cdots
j_n}\left(\omega_1,\cdots,\omega_n\right)$. The quantum
susceptibility coefficients of the total system can be calculated
by
\begin{eqnarray}\label{The
first system in the series product is linear system}
\lefteqn{\chi_{j,i_1\cdots i_n}^{i,j_1\cdots
j_n}\left(\omega_1,\cdots,\omega_n\right)}\nonumber\\
&&{\hspace{0.5cm}}=\chi_{j,i_1\cdots
i_n}^{\left(2\right)\,\,i,j_1\cdots
j_n}\left(\omega_1,\cdots,\omega_n\right)\prod_{i=1}^n
G^{\left(1\right)}\!\left(i\omega_i\right).
\end{eqnarray}

(ii) The input field is first fed into a nonlinear component with
quantum susceptibility coefficients $\chi_{j,i_1\cdots
i_n}^{\left(1\right)\,\,i,j_1\cdots
j_n}\left(\omega_1,\cdots,\omega_n\right)$, and then guided into a
linear component with quantum transfer function
$G^{\left(2\right)}\!\left(s\right),\,s\in\mathbb{C}$. The quantum
susceptibility coefficients of the series-product system can be
expressed as
\begin{eqnarray}\label{The
second system in the series product is linear system}
\lefteqn{\chi_{j,i_1\cdots i_n}^{i,j_1\cdots
j_n}\left(\omega_1,\cdots,\omega_n\right)}
 \nonumber\\
&&{\hspace{0.5cm}}=G^{\left(2\right)}\!\left(i\sum_{i=1}^n
\omega_i\right)\chi_{j,i_1\cdots
i_n}^{\left(1\right)\,\,i,j_1\cdots
j_n}\left(\omega_1,\cdots,\omega_n\right).
\end{eqnarray}

Note that the above examples are quite useful for modelling a
large class of important quantum input-output networks, such as
the network with a nonlinear component cascaded connected to a
quantum amplifier, that cannot be modelled appropriately by the
existing approaches.

\section{Applications}\label{s5}

The Volterra series approach we introduce here can be applied to
various linear and nonlinear quantum input-output systems,
especially those with weak nonlinearity. To show this, we study
the input-output relation of some conventional nonlinear
components, which can be taken as the basic elements of more
complex quantum networks.

\begin{example}(Kerr Cavity)

As a first example, we consider a Kerr cavity with free
Hamiltonian $H=\omega_a
a^{\dagger}a+\chi\left(a^{\dagger}a^{\dagger} a a\right)$ and
dissipation operator $L=\sqrt{\gamma}a$ coupled to the input
field, where $a$ and $a^{\dagger}$ are the annihilation and
creation operators of the cavity. Here $\omega_a$, $\chi$, and
$\gamma$ are the frequency of the fundamental mode, the nonlinear
Kerr coefficient of the cavity, and the coupling strength between
the cavity and the input field (see Fig.~\ref{Fig of the Kerr
cavity}).
\begin{figure}[t]
\centerline{\includegraphics[width=7.6 cm]{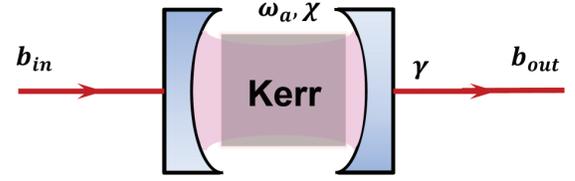}}
\caption{(color online) Schematic diagram of an nonlinear
input-output Kerr cavity with frequency $\omega_a$, nonlinear Kerr
coefficient $\chi$, and damping rate $\gamma$.}\label{Fig of the
Kerr cavity}
\end{figure}
Let us consider the weak-nonlinear assumption such that
$\chi\ll\omega_,\,\gamma$, then from theorem~\ref{Volterra series
for weak-nonlinear system} we can expand the quantum Volterra
series up to the third-order terms. If we further assume that the
cavity is initially in the vacuum state, there is only one nonzero
first-order Volterra kernel
\begin{equation}\label{First-order Volterra kernel of Kerr cavity}
k_-\left(\tau\right)=-\gamma\exp\left[-\left(\frac{\gamma}{2}+i\omega_a\right)\tau\right]
\end{equation}
and four nonzero third-order Volterra kernels
\begin{eqnarray}\label{Third-order Volterra kernels of Kerr cavity}
\lefteqn{k_{\pm-+}\left(\tau_1,\tau_2,\tau_3\right)}\nonumber\\
&&{\hspace{0.3cm}}=\frac{4i\gamma^2\chi^2}{-\gamma+i\chi}e^{-\gamma/2\left(\tau_1+\tau_3\right)-i\omega_a\left(\tau_1-\tau_3\right)-\gamma
\tau_2}\left(1-e^{-\gamma\tau_1}\right),\nonumber\\
\lefteqn{k_{\pm+-}\left(\tau_1,\tau_2,\tau_3\right)}\nonumber\\
&&{\hspace{0.3cm}}=-\frac{4i\gamma^2\chi^2}{-\gamma+i\chi}e^{-\left(\gamma/2+i\omega_a\right)\left(\tau_1+\tau_3\right)-\gamma\tau_2}\left(1-e^{-\gamma\tau_1}\right).\nonumber\\
\end{eqnarray}
See the derivations in the appendix.

For this example, the Volterra series approach gives a more exact
description of the quantum input-output response compared with
other approximation approaches, such as the truncation
approximation approach in the Fock space which is mainly used for
low-excitation quantum systems and the semiclassical
approximation, which is traditionally introduced to study
highly-excited systems. To show this, let us see the simulation
results given in Fig.~\ref{Fig of the evolution of the normalized
position operator}. Given the system parameters
$\left(\chi/\omega_a,\gamma/\omega_a\right)=\left(0.01,0.2\right)$,
the output trajectory obtained by the quantum Voterra series
approach matches more perfectly well with the ideal trajectory
compared with those obtained by the few-photon truncation in the
Fock space and the semiclassical approximation.
\begin{figure}[t]
\centerline{\includegraphics[width=7.6 cm]{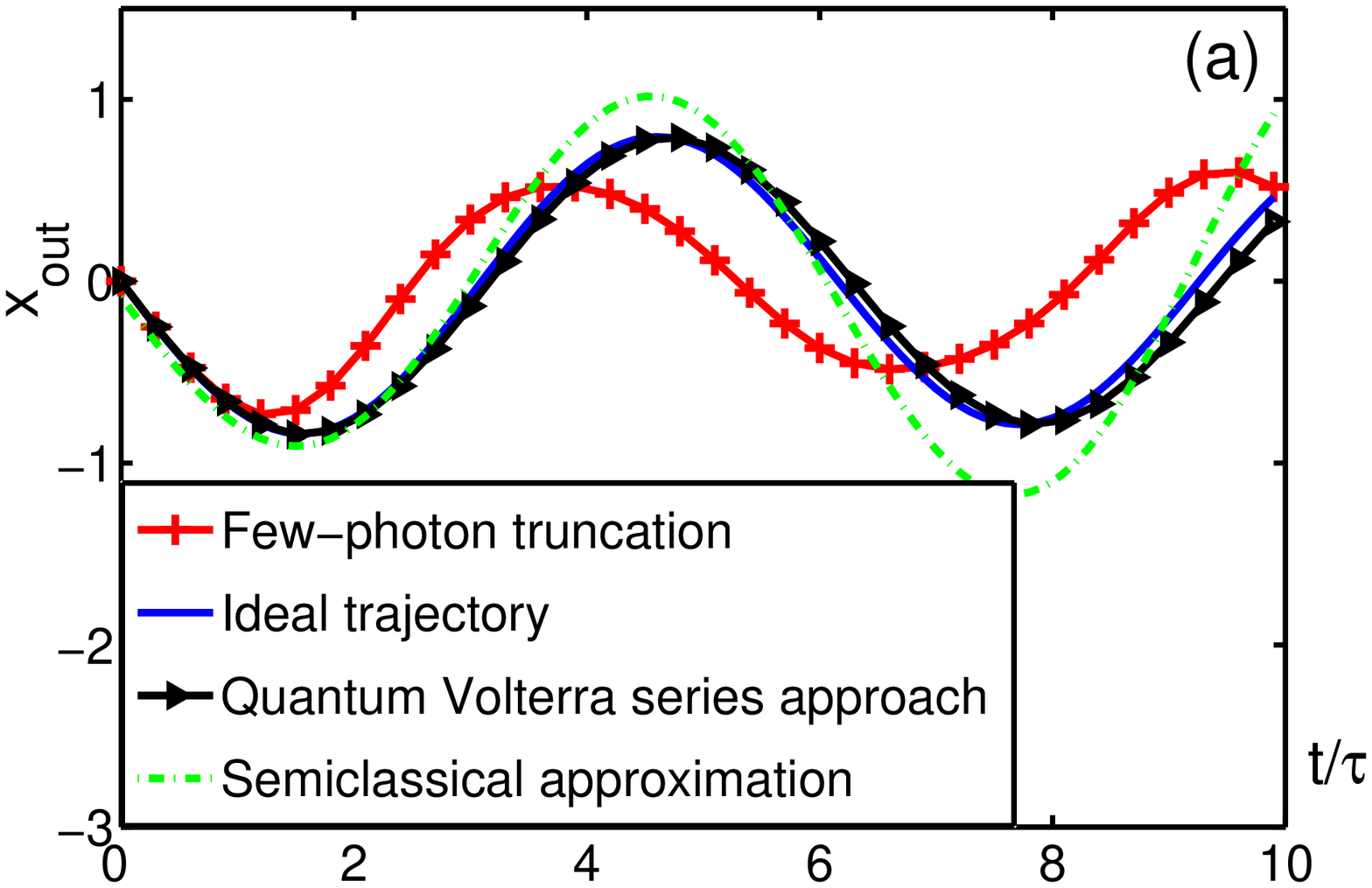}}
\centerline{\includegraphics[width=7.6 cm]{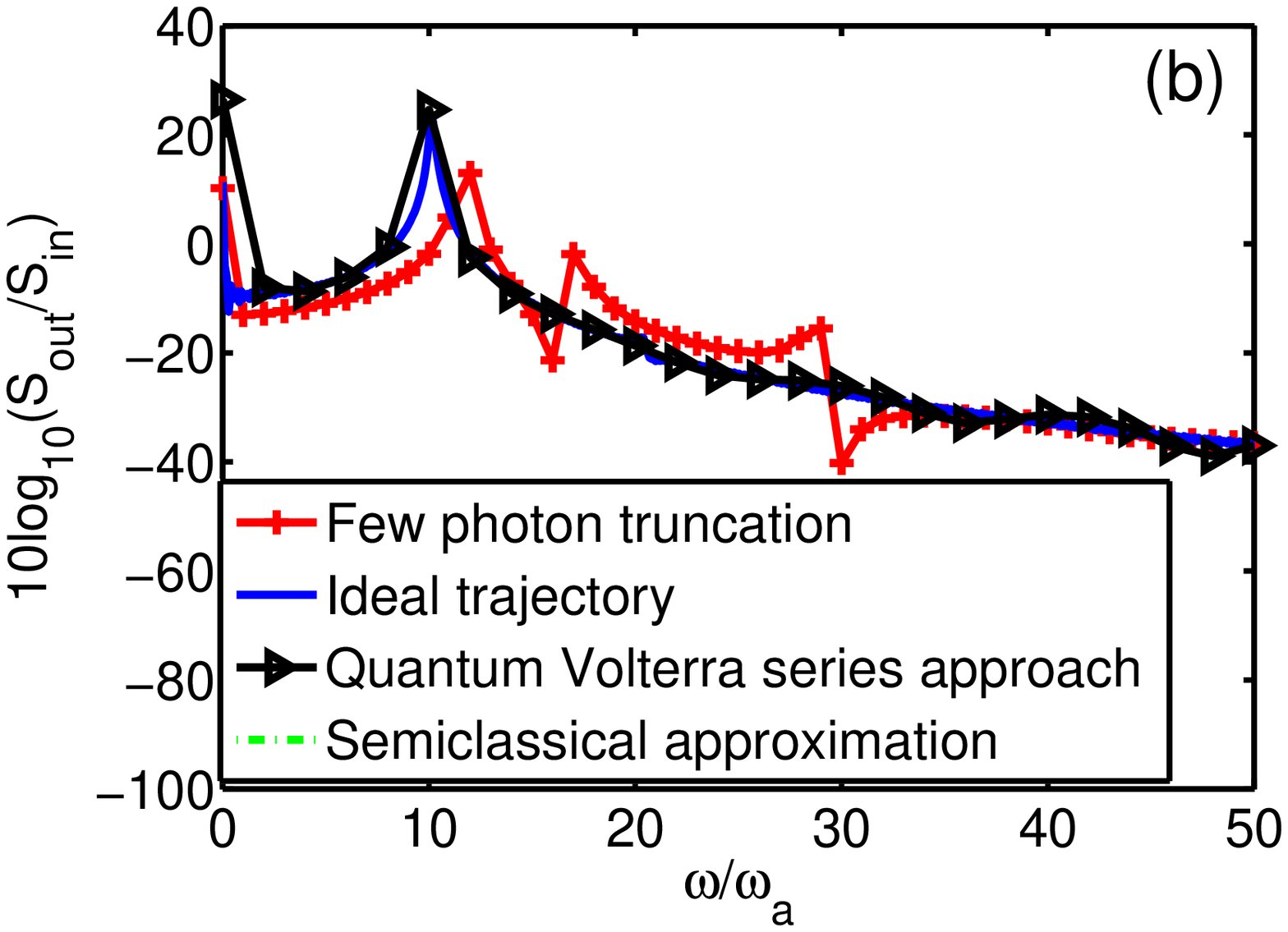}}
\caption{(color online) (a) Time evolution of the output field
$x_{\rm out}=\left(b_{\rm out}+b_{\rm
out}^{\dagger}\right)/\sqrt{2}$ and (b) logarithmic output spectra
for a Kerr cavity with
$\left(\chi/\omega_a,\gamma/\omega_a\right)=\left(0.01,0.2\right)$.
In order to obtain the output spectrum, we drive the Kerr cavity
by an external field with strength $\epsilon_d=0.6\,\omega_a$.
Here $\tau=2\pi/\omega_a$ is a normalized unit of time. The blue
solid curve is the ideal trajectory. The black triangle curve, the
green dashed curve, and the red curve with plus signs are the
trajectories obtained by the Volterra series approach,
semiclassical approximation, and the few-photon truncation with
expansion up to five-photon Fock state. The trajectory obtained by
the Volterra series approach coincides very well with the ideal
one compared with the other two approaches.}\label{Fig of the
evolution of the normalized position operator}
\end{figure}

\end{example}
\vspace{6pt}

\begin{example}(Optomechanical transducer)

In the second example, let us concentrate on a single-mode cavity
parametrically coupled to a mechanical oscillator (see
Fig.~\ref{Fig of the optomechanical transducer}), which received a
high degree of attention
recently~\cite{MPoot,GJMilburn,RHamerlyPRL:2012,XYLvSR:2013,JQLiaoPRA:2012,XWXuPRA:2013}.
These systems can be used as sensitive detectors to detect spin
and mass, or a sensitive mechanical transducer. The Hamiltonian of
this system can be written as $H=\omega_a a^{\dagger}a+\omega_b
b^{\dagger}b +ga^{\dagger}a(b+b^{\dagger})$, where $a$
($a^{\dagger}$) and $b$ ($b^{\dagger}$) are the annihilation
(creation) operators of the cavity mode and the mechanical
oscillator; $\omega_a$ and $\omega_b$ are the angular frequencies
of these two modes; and $g$ is the optomechanical coupling
strength. The cavity mode is coupled to the input field with
damping rate $\gamma_a$, and $\gamma_b$ is the damping rate of the
mechanical mode.
\begin{figure}[t]
\centerline{\includegraphics[width=7.6 cm]{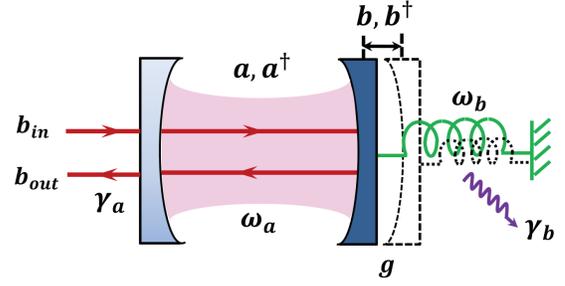}}
\caption{(color online) Schematic diagram of an optomechanical
transducer where $a$, $\omega_a$, and $\gamma_a$ ($b$, $\omega_b$,
and $\gamma_b$) are the annihilation operators of the optical mode
and the mechanical mode. $g$ is the optomechanical coupling
strength.}\label{Fig of the optomechanical transducer}
\end{figure}
Let us assume that the optomechanical coupling is weak enough such
that $g\ll\omega_{a,b},\,\gamma_{a,b}$. Note that $g$ determines
the nonlinearity of the optomechanical systems, thus the above
assumption means that the nonlinearity of the optomechanical
system we consider is weak. From theorem~\ref{Volterra series for
weak-nonlinear system}, we can expand the quantum Volterra series
to the third-order terms and omit higher-order terms. If we
further assume that the cavity and the mechanical oscillator are
both initially in the vacuum states and note that
$\omega_a\gg\omega_b$ and $\gamma_a\gg\gamma_b$, there are only
one non-zero first-order Volterra kernel
$k_-\left(\tau\right)=-\gamma_a\exp\left[-\left(\frac{\gamma_{a}}{2}+i\omega_a\right)\tau\right]$
and two non-zero third-order Volterra kernels
\begin{eqnarray}\label{Third-order Volterra kernels of optomechanical components}
\lefteqn{k_{-\pm\mp}\left(\tau_1,\tau_2,\tau_3\right)=}\nonumber\\
&&\gamma^2_a g^2
e^{-\gamma^a_{\pm}\tau_3}\left(\frac{e^{-\gamma^a_{-}\tau_1}+e^{-\gamma^a_{+}\tau_1}}{-\gamma^a_{-}+\gamma^a_{+}}\right)
\left(\frac{e^{-\gamma_a\tau_2}-e^{-\gamma^b_{-}\tau_2}}{-\gamma^b_{-}+\gamma_a}\right),\nonumber\\
\end{eqnarray}
where $\gamma^a_{\pm}=\frac{\gamma_a}{2}\pm i\omega_a$ and
$\gamma^b_-=\frac{\gamma_b}{2}-i \omega_b$. The derivations of
Eq.~(\ref{Third-order Volterra kernels of optomechanical
components}) are similar to those of Eqs.~(\ref{First-order
Volterra kernel of Kerr cavity}) and (\ref{Third-order Volterra
kernels of Kerr cavity}) given in the appendix, thus we omit those
here.

\end{example}
\vspace{6pt}

\begin{example} (Nonlinear coherent feedback network with weak Kerr nonlinearity)

The Volterra series approach gives a simpler way to analyze
nonlinear quantum coherent feedback control
systems~\cite{JZhang2,RHamerlyPRL:2012,HMabuchiAPL:2011,JKerckhoffPRL:2012,JKerckhoffPRL:2010,ZZhouAPL:2012,ZPLiuPRA:2013}.
To show this, let us consider a simple coherent feedback system in
Fig.~\ref{Fig of the evolution of the normalized position
operator}(a). In this system, the controlled system is a linear
cavity $S_1$, and in the feedback loop there are a quantum
amplifier $S_2$ and a Kerr nonlinear component $S_3$. Notice that
this coherent feedback system cannot be modelled by the existing
approaches such as the Hudson-Parthasarathy model and the quantum
transfer function model, but we can describe it by our approach.
The total system can be seen as a cascade-connected system $S=S_1
\vartriangleleft S_3 \vartriangleleft S_2 \vartriangleleft S_1$.
The system dynamics can be obtained from Eqs.~(\ref{The first
system in the series product is linear system}), (\ref{The second
system in the series product is linear system}), and example 1.
This quantum coherent feedback loop induces an interesting
phenomenon: the nonlinear component in the coherent feedback loop
changes the dynamics of the linear cavity and make it a nonlinear
cavity. This nonlinear effect is additionally amplified by the
quantum amplifier in the feedback loop. Let $\omega_a$, $\gamma_a$
be the effective frequency and damping rate of the controlled
linear cavity. $\omega_b$, $\chi_b$, $\gamma_b$ are the effective
frequency, nonlinear Kerr coefficient, and damping rate of the
nonlinear Kerr cavity in the feedback loop. $G$ is the power gain
of the quantum amplifier in the feedback loop. Under the condition
that $\gamma_a\gg\omega_a$, the controlled cavity can be seen as a
nonlinear Kerr cavity with effective Kerr coefficient
$\tilde{\chi}_a=G\chi_b$. This amplified nonlinear Kerr effects
leads to nonlinear quantum phenomena in the controlled cavity. For
example, if the initial state of the controlled cavity is a
coherent state, this state will evolve into a non-Gaussian state,
which is highly nonclassical. In Fig.~\ref{Fig of the evolution of
the normalized position operator}(b), we use the measure
$$\delta\left[\rho\right]=\frac{{\rm
tr}\left[(\rho-\sigma)^2/2\right]}{{\rm
tr}\left[\rho^2\right]}\in[0,1]$$ to evaluate the non-Gaussian
degree of the quantum states generated in the controlled
cavity~\cite{Genoni}, where $\sigma$ is a Gaussian state with the
same first and second-order quadratures of the non-Gaussian state
$\rho$. Simulation results in Fig.~\ref{Fig of the evolution of
the normalized position operator}(b) show that higher-quality
non-Gaussian states can be obtained if we increase the power gain
$G$ of the quantum amplifier in the feedback loop. We should point
out that we have predicted a similar quantum feedback
nonlinearization phenomenon in Ref.~\cite{JZhang2}. But in that
paper, the nonlinearity is induced by the nonlinear dissipation
interaction between the controlled system and the mediated quantum
field, and the feedback loop is linear. Here, we show that
nonlinear coherent feedback loop can induce quantum
nonlinearity~\cite{ZPLiuPRA:2013}, which can be further amplified
by the quantum amplifier in the feedback loop.
\begin{figure}[t]
\centerline{\includegraphics[width=7.6 cm]{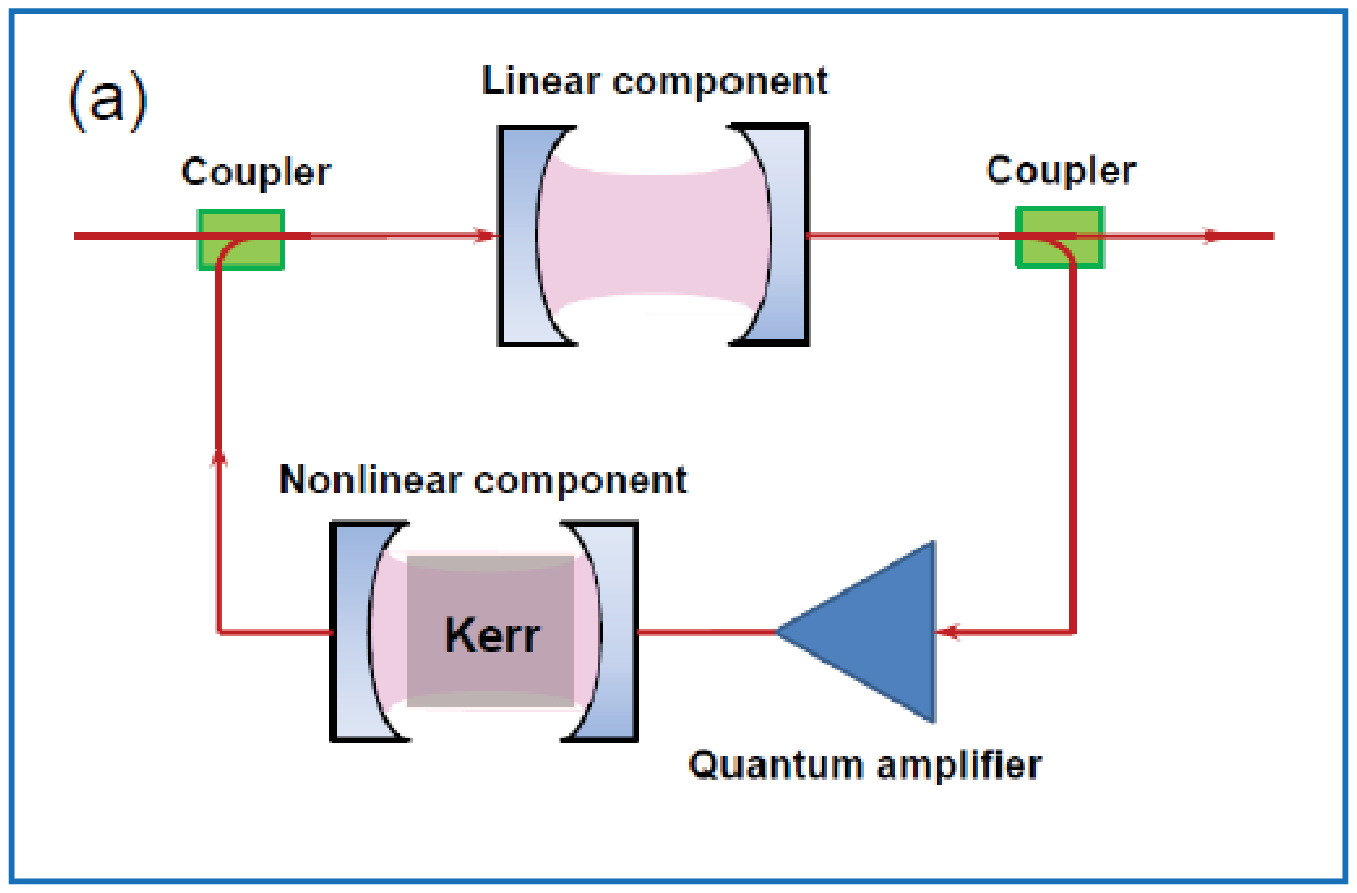}}
\centerline{\includegraphics[width=7.6 cm]{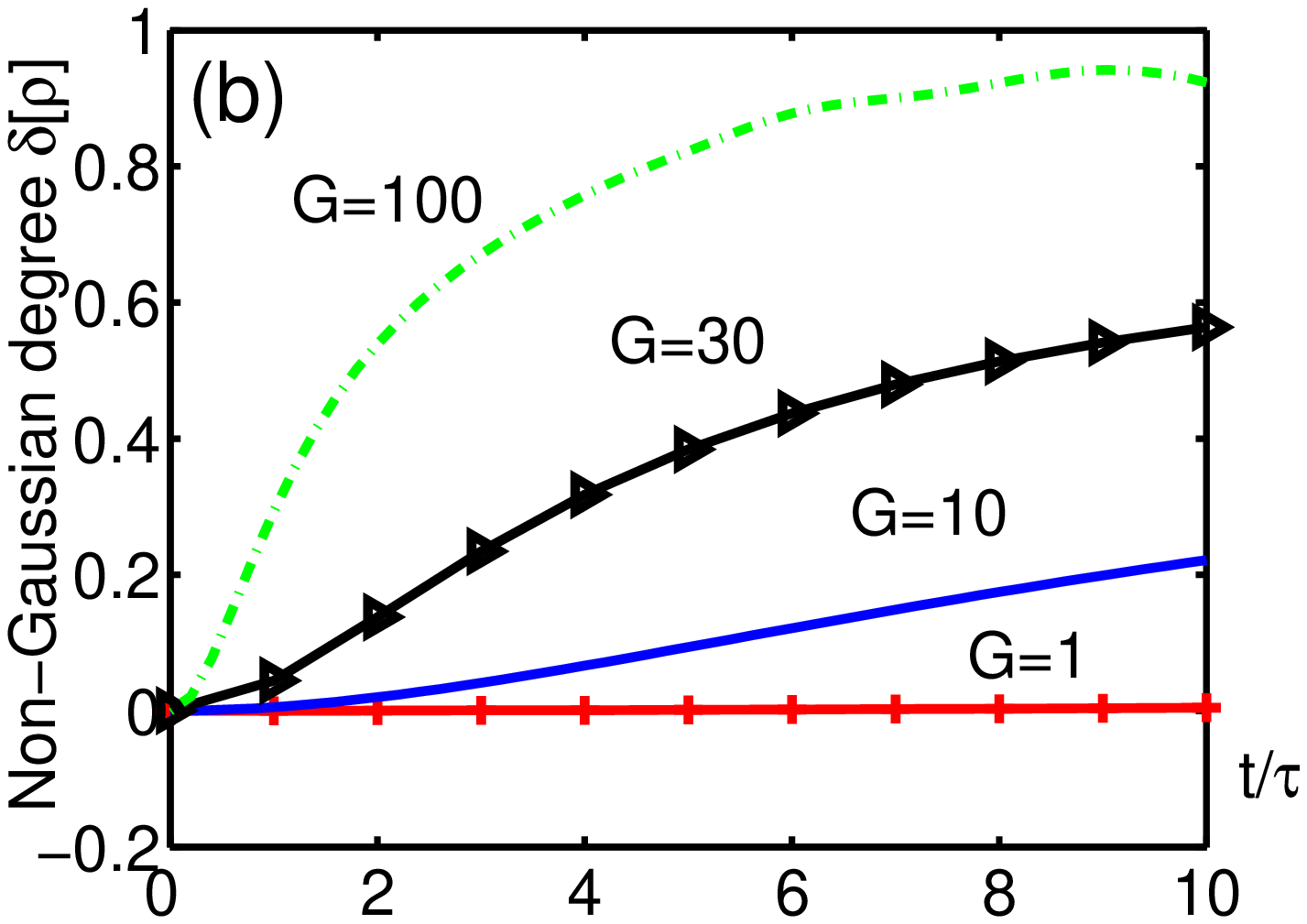}}
\caption{(color online) Quantum feedback nonlinearization: (a) the
schematic diagram of the nonlinear coherent feedback loop; (b) the
non-Gaussian degree of the quantum state in the controlled
cavity.}\label{Fig of the evolution of the normalized position
operator}
\end{figure}

\end{example}

\section{Conclusion}\label{s6}


%
%

We have introduced a new formalism of quantum input-output
networks using the so-called Volterra series. It gives a simpler
way to describe large-scale nonlinear quantum input-output
networks especially in the frequency domain, and can be also used
to analyze more general quantum networks with both nonlinear
components and quantum amplifiers that cannot be modelled by the
existing methods such as the Hudson-Parthasarathy model and the
quantum transfer function model. An application to quantum
coherent feedback systems shows that it can be used to show the
quantum feedback nonlinearization effects, in which the nonlinear
components in the coherent feedback loop can change the dynamics
of the controlled linear system and these quantum nonlinear
effects can be amplified by a linear quantum amplifier. Our work
opens up new perspectives in nonlinear quantum networks,
especially quantum coherent feedback control systems.

\appendix
{\it Proof of the theorem~\ref{Volterra series for weak-nonlinear
system}:} Let us assume that ${\bf X}_1$ is an operator vector of
which the components are linear terms of the annihilation and
creation operators $a_i$ and $a_i^{\dagger}$, $i=1,\cdots,r$, and
${\bf X}_2$ is another operator vector of which the components are
higher-order nonlinear terms of $a_i$ and $a_i^{\dagger}$. The
system dynamics can be fully determined by the vector ${\bf
X}=\left({\bf X}_1^T,{\bf X}_2^T\right)^T$. From Eq.~(\ref{Quantum
stochastic differential equation}) and the special form of the
system Hamiltonian $H$ and dissipation operator ${\bf L}$ in
Eq.~(\ref{Hamiltonian and dissipation operator of weak-nonlinear
quantum system}), we can obtain the dynamical equations of ${\bf
X}_1$ and ${\bf X}_2$ as follows
\begin{eqnarray}
\dot{{\bf X}}_1&=& \left( A_{11}{\bf X}_1 +\mu A_{12}{\bf X}_2
\right)+\left(B_1^*{\bf
b}_{\rm in}+B_1{\bf b}_{\rm in}^{\dagger}\right),\label{X1}\\
\dot{{\bf X}}_2&=& \left(A_{21}{\bf X}_1+A_{22}{\bf
X}_2\right)+\left(B_{21}^*{\bf b}_{\rm in}+B_{21}{\bf b}_{\rm
in}^{\dagger}\right){\bf X}_1\nonumber\\
&&+\mu\left(B_{22}^*{\bf b}_{\rm in}+B_{22}{\bf b}_{\rm
in}^{\dagger}\right){\bf X}_2,\label{X2}
\end{eqnarray}
where $A_{ij}$, $B_{ij}$ are constant matrices determined by
$H_l$, $H_{nl}$, ${\bf L}_l$, and ${\bf L}_{nl}$. By solving
Eqs.~(\ref{X1}) and (\ref{X2}), we have
\begin{eqnarray}
{\bf X}_1\!\left(t\right)&=&e^{A_{11}t}{\bf X}_1+\mu\int_0^t\!\!
e^{A_{11}\left(t-\tau\right)}A_{12}{\bf
X}_2\!\left(\tau\right)d\tau\nonumber\\
&&+\int_0^t\!\! e^{A_{11}\left(t-\tau\right)}\left[B_1^*{\bf
b}_{\rm in}\!\left(\tau\right)+B_1{\bf b}_{\rm
in}^{\dagger}\!\left(\tau\right)\right]d\tau,\nonumber\\ \label{X1 solution}\\
{\bf X}_2\!\left(t\right)&=&e^{A_{22}t}{\bf X}_2+\int_0^t\!\!
e^{A_{22}\left(t-\tau\right)}A_{21}{\bf
X}_1\!\left(\tau\right)d\tau\nonumber\\
&&+\int_0^t\!\!\left[B_{21}^*{\bf b}_{\rm
in}\!\left(\tau\right)+B_{21}{\bf b}_{\rm
in}^{\dagger}\!\left(\tau\right)\right]X_1\!\left(\tau\right)d\tau\nonumber\\
&&+\mu\int_0^t\!\!\left[B_{22}^*{\bf b}_{\rm
in}\!\left(\tau\right)+B_{22}{\bf b}_{\rm
in}^{\dagger}\!\left(\tau\right)\right]X_2\!\left(\tau\right)d\tau.\nonumber\\
\label{X2 solution}
\end{eqnarray}
Noticing that
\begin{eqnarray*}
{\bf b}_{\rm out}\!\left(t\right)={\bf b}_{\rm
in}\!\left(t\right)+{\bf l}_1\langle
X_1\!\left(t\right)\rangle_0+\mu {\bf l}_2\langle
X_2\!\left(t\right)\rangle_0,
\end{eqnarray*}
we can obtain Eq.~(\ref{Volterra series of weak-nonlinear quantum
input-output system}) by iterating Eq.~(\ref{X2 solution}) into
Eq.~(\ref{X1 solution}).

{\it Derivations of Eqs.~(\ref{First-order Volterra kernel of Kerr
cavity}) and (\ref{Third-order Volterra kernels of Kerr cavity}):}
Under the weak nonlinearity assumption $\chi\ll\omega_a,\gamma$,
we can expand the system dynamics of the Kerr cavity up to the
third-order terms of $a$ and $a^{\dagger}$ and omit higher-order
terms. Then, the system dynamics can be expressed as the following
quantum stochastic differential equation
\begin{eqnarray*}
\dot{\bf X}=A{\bf X}+B_-{\bf X}b_{\rm in}+B_+{\bf X}b_{\rm
in}^{\dagger},
\end{eqnarray*}
where ${\bf
X}=\left(I,a,a^{\dagger},a^{\dagger}a,a^2,a^{\dagger\,2},a^{\dagger}a^2,aa^{\dagger\,2}\right)^T$.
The coefficient matrices $A$, $B_-$, and $B_+$ can be given by
Eq.~(\ref{Coefficient matrices of Kerr cavity}), where
$\tilde{\omega}_a=\omega_a-\chi$.
\begin{figure*}
\begin{eqnarray}\label{Coefficient matrices of Kerr cavity}
&A=\left(%
\begin{array}{cccccccc}
  0 & 0 & 0 & 0 & 0 & 0 & 0 & 0 \\
  0 & -\left(\frac{\gamma}{2}+i\omega_a\right) & 0 & 0 & 0 & 0 & -2i\chi & 0 \\
  0 & 0 & -\left(\frac{\gamma}{2}-i\omega_a\right) & 0 & 0 & 0 & 0 & 2i\chi \\
  0 & 0 & 0 & -\gamma & 0 & 0 & 0 & 0 \\
  0 & 0 & 0 & 0 & -\left(\gamma+2i\omega_a\right) & 0 & 0 & 0 \\
  0 & 0 & 0 & 0 & 0 & -\left(\gamma-2i\omega_a\right) & 0 & 0 \\
  0 & 0 & 0 & 0 & 0 & 0 & -\left[\frac{3\gamma}{2}+i\tilde{\omega}_a\right] & 0 \\
  0 & 0 & 0 & 0 & 0 & 0 & 0 & -\left[\frac{3\gamma}{2}-i\tilde{\omega}_a\right] \\
\end{array}%
\right),&\nonumber\\
&B_-=\left(%
\begin{array}{cccccccc}
  0 & 0 & 0 & 0 & 0 & 0 & 0 & 0 \\
  -\sqrt{\gamma} & 0 & 0 & 0 & 0 & 0 & 0 & 0 \\
  0 & 0 & 0 & 0 & 0 & 0 & 0 & 0 \\
  0 & 0 & -\sqrt{\gamma} & 0 & 0 & 0 & 0 & 0 \\
  0 & -2\sqrt{\gamma} & 0 & 0 & 0 & 0 & 0 & 0 \\
  0 & 0 & 0 & 0 & 0 & 0 & 0 & 0 \\
  0 & 0 & 0 & 2\sqrt{\gamma} & 0 & 0 & 0 & 0 \\
  0 & 0 & 0 & 0 & 0 & -\sqrt{\gamma} & 0 & 0 \\
\end{array}%
\right),&\nonumber\\
&B_+=\left(%
\begin{array}{cccccccc}
  0 & 0 & 0 & 0 & 0 & 0 & 0 & 0 \\
  0 & 0 & 0 & 0 & 0 & 0 & 0 & 0 \\
  -\sqrt{\gamma} & 0 & 0 & 0 & 0 & 0 & 0 & 0 \\
  0 & -\sqrt{\gamma} & 0 & 0 & 0 & 0 & 0 & 0 \\
  0 & 0 & 0 & 0 & 0 & 0 & 0 & 0 \\
  0 & 0 & -2\sqrt{\gamma} & 0 & 0 & 0 & 0 & 0 \\
  0 & 0 & 0 & 2\sqrt{\gamma} & 0 & 0 & 0 & 0 \\
  0 & 0 & 0 & 0 & -\sqrt{\gamma} & 0 & 0 & 0 \\
\end{array}%
\right).&
\end{eqnarray}
\end{figure*}
With a similar iteration process as the one given by Eq.~(\ref{Iterative
expansion of the operator vector}), we can calculate the Volterra
kernels using
\begin{eqnarray*}
k_-\left(\tau\right)=l^T\exp\left(A\tau\right)B_-x_0,
\end{eqnarray*}
and
\begin{eqnarray*}
k_{\pm\pm\pm}(\tau_1,\tau_2,\tau_3)=l^Te^{A\tau_1}B_{\pm}e^{A\tau_2}B_{\pm}e^{A\tau_3}B_{\pm}x_0,
\end{eqnarray*}
where $l=\left(010\cdots0\right)^T$ and $x_0=\langle {\bf X}
\rangle_0=\left(10\cdots0\right)^T$. $\langle\cdot\rangle_0$ is
the average taking over the initial vacuum state of the Kerr
cavity.

\section*{Acknowledgment}
J. Zhang would like to thank Prof. G.~F. Zhang for his helpful comments.

\begin{biography}[{\includegraphics[width=1in,height=1.25in,clip,keepaspectratio]{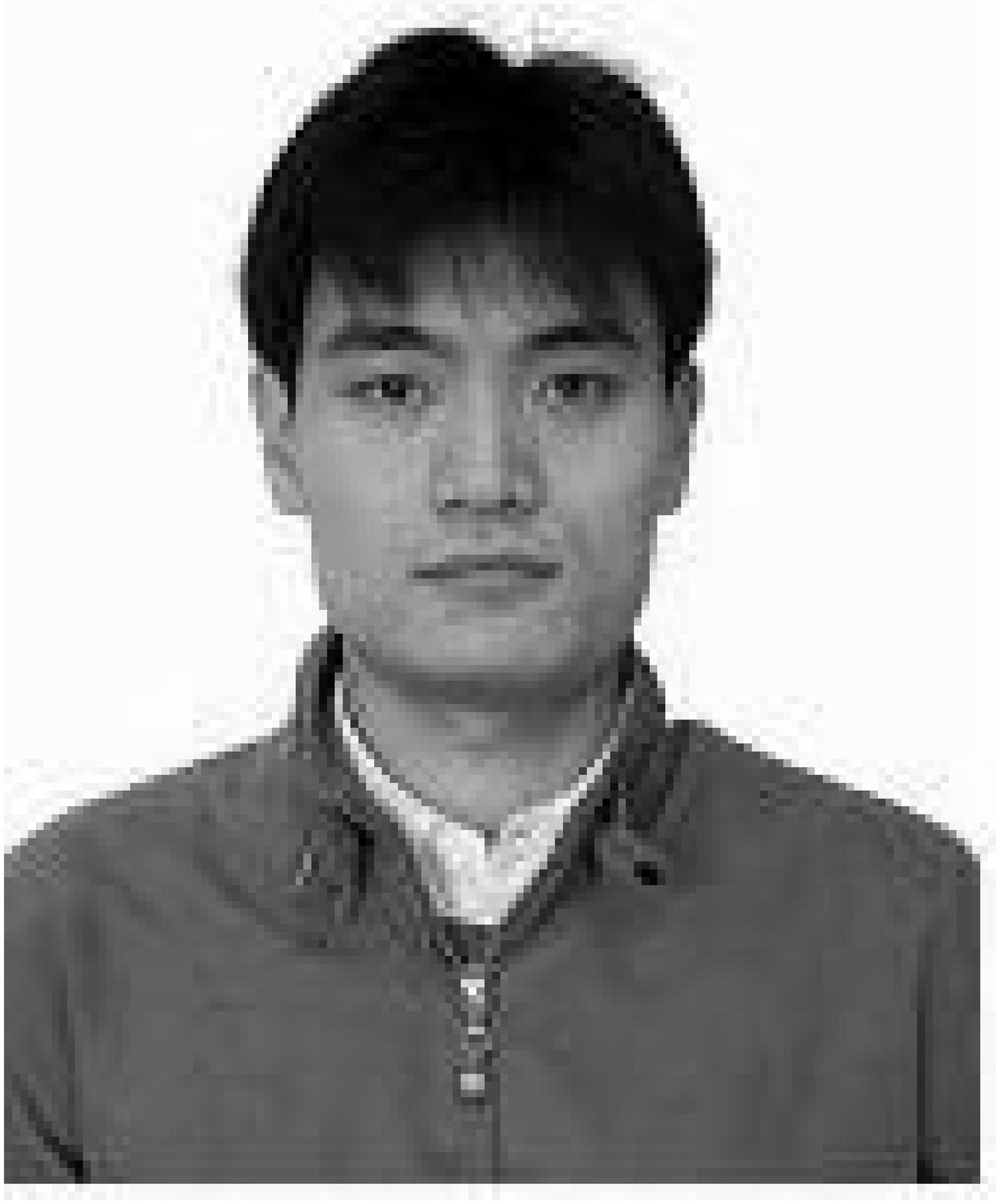}}]{Jing Zhang}
received his B.S. degree from Department of Mathematical Science
and Ph.D. degree from Department of Automation, Tsinghua
University, Beijing, China, in 2001 and 2006, respectively.

From 2006 to 2008, he was a Postdoctoral Fellow at the Department
of Computer Science and Technology, Tsinghua University, Beijing,
China, and a Visiting Researcher from 2008 to 2009 at the Advanced
Science Institute, the Institute of Physical and Chemical Research
(RIKEN), Japan. In 2010, he worked as a Visiting Assistant
Professor at Department of Physics and National Center for
Theoretical Sciences, National Cheng Kung University, Taiwan. He
is now an Associate Professor at the Department of Automation,
Tsinghua University, Beijing, China. His research interests
include quantum control and nano manipulation.
\end{biography}

\begin{biography}[{\includegraphics[width=1in,height=1.25in,clip,keepaspectratio]{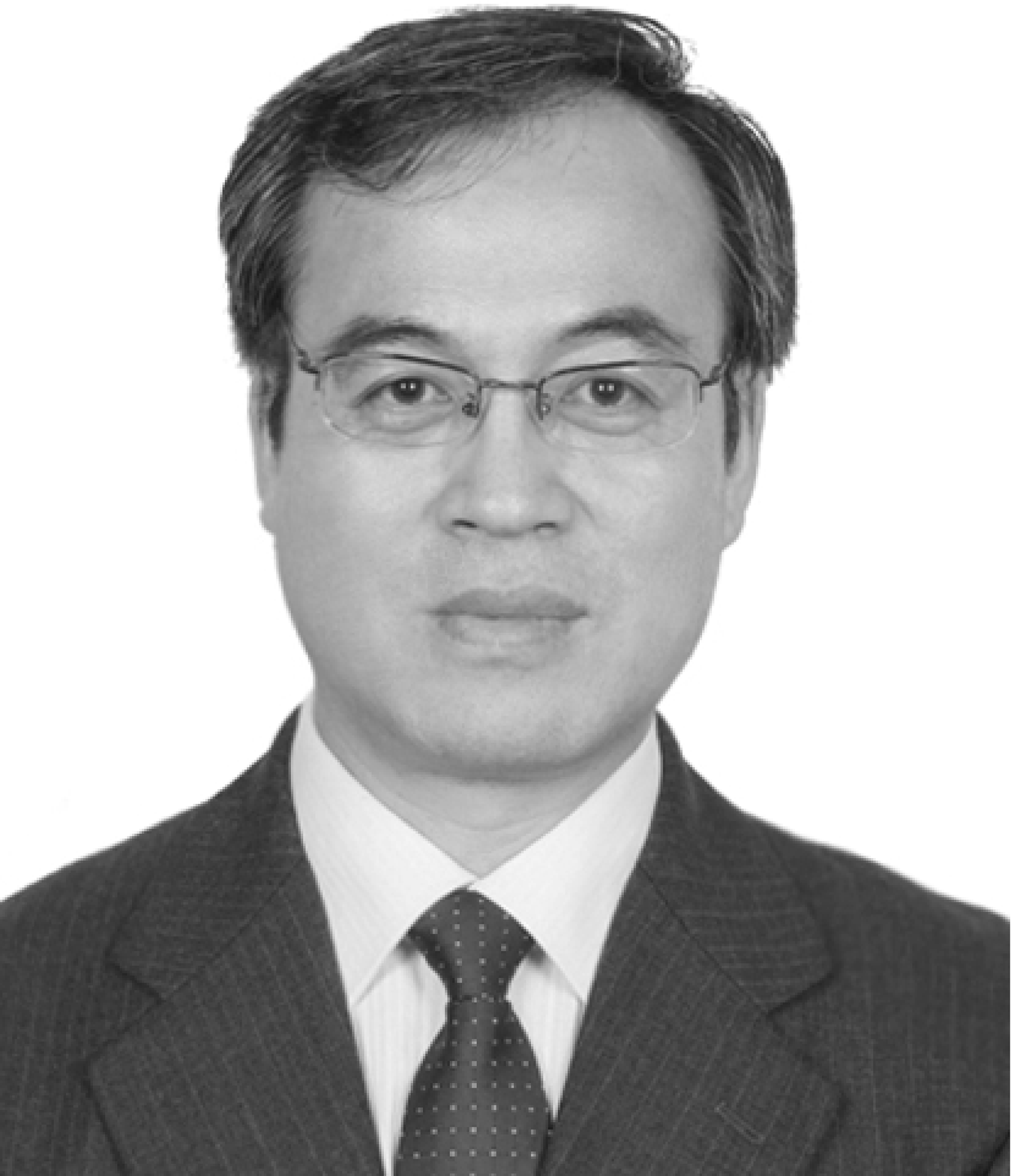}}]{Yu-xi Liu}
received his B.S. degree, M.S. degree and Ph.D. degree from
Department of Physics, Shanxi Normal University, Jilin University
and Peking University in 1989, 1995 and 1998, respectively.

From 1998 to 2000, he was a Post-doctor at the Institute of
Theoretical Physics, the Chinese Academy of Sciences, China. From
2000 to 2002, he was a JSPS Postdoctoral fellow at the Graduate
University for Advanced Studies (SOKENDAI), Japan. From 2002 to
2009, he was a research scientist in the Institute of Physical and
Chemical Research (RIKEN), Japan.

Since 2009, he has been a Professor with Institute of
Microelectronics, Tsinghua University. His research interests
include solid state quantum devices, quantum information
processing, quantum optics and quantum control theory.
\end{biography}

\begin{biography}[{\includegraphics[width=1in,height=1.25in,clip,keepaspectratio]{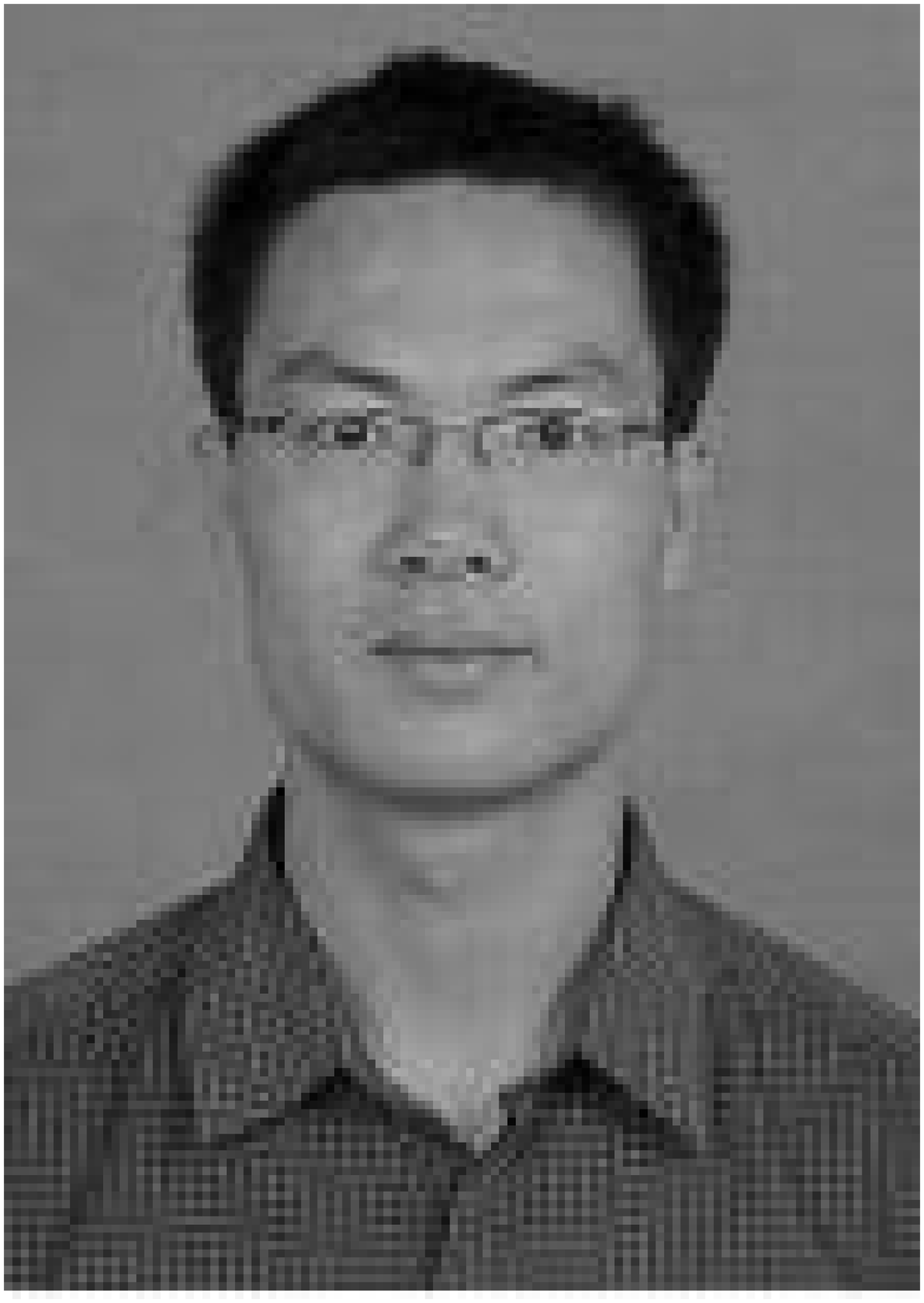}}]{Re-Bing Wu}
received his B.S. degree in Electrical Engineering and Ph.D.
degree in Control Science and Engineering from Tsinghua
University, Beijing, China, in 1998 and 2004, respectively.

From 2005 to 2008, he was a Research Associate Fellow at the
Department of Chemistry, Princeton University, USA. Since 2009, he
has been an Associate Professor at the Department of Automation,
Tsinghua University, Beijing, China. His research interests
include quantum mechanical control theory and nonlinear control
theory.
\end{biography}

\begin{biography}[{\includegraphics[width=1in,height=1.25in,clip,keepaspectratio]{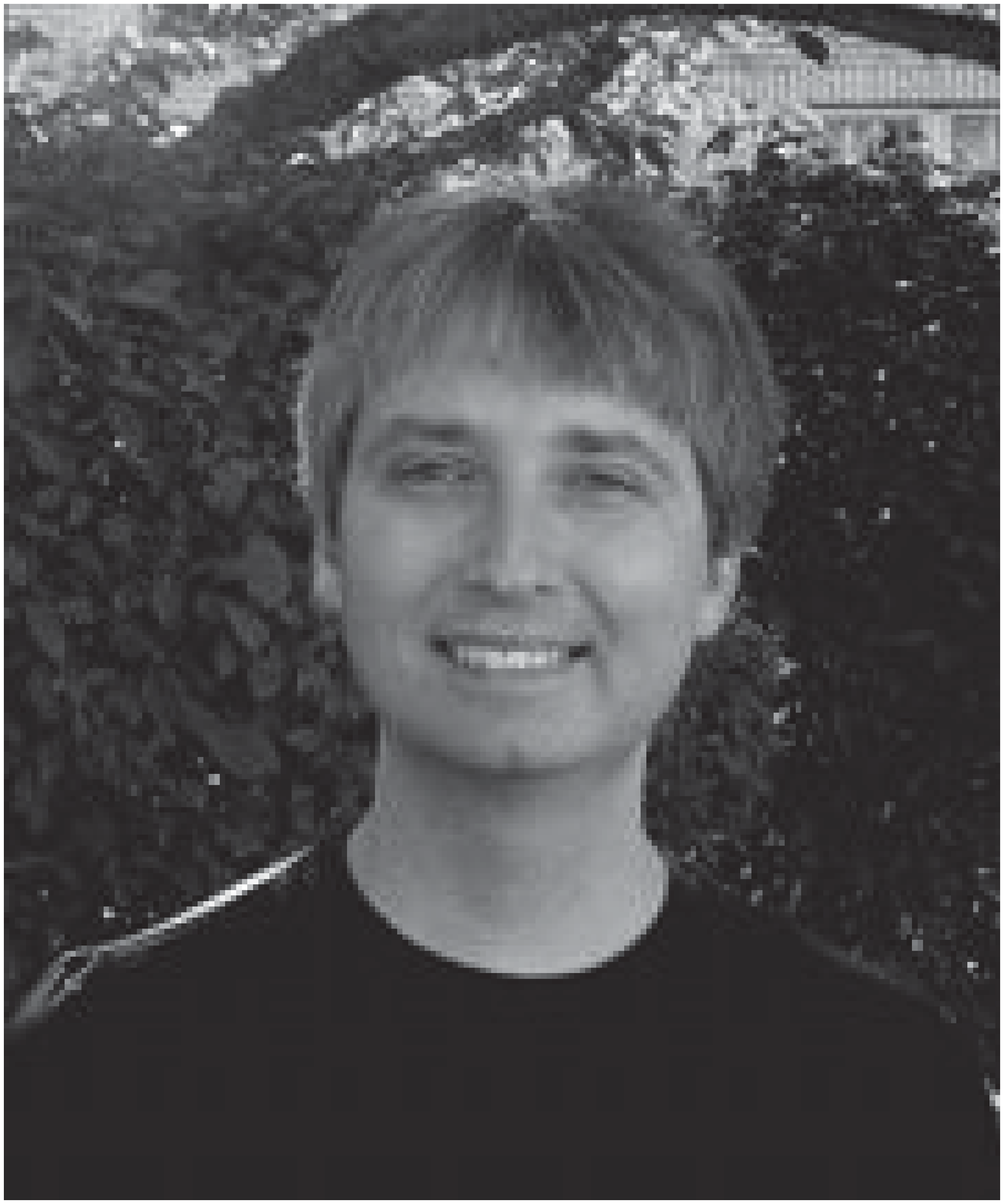}}]{Kurt Jacobs}
received a B.S. and M.S. degree in Physics from the University of Auckland, and a Ph.D. degree in Physics from Imperial College, London,
which he completed in 1998. He then held postdoctoral positions at Los Alamos National Laboratory, Griffith University, and Louisiana State University,
before joining the University of Massachusetts at Boston as an Assistant Professor. He has been an Associate Professor there since 2011.
His research interests include quantum measurement theory, feedback control in mesoscopic systems, and quantum thermodynamics. He is the author of Stochastic Processes for Physicists, and Quantum Measurement Theory and Its Applications, from Cambridge University Press.
\end{biography}

\begin{biography}[{\includegraphics[width=1in,height=1.25in,clip,keepaspectratio]{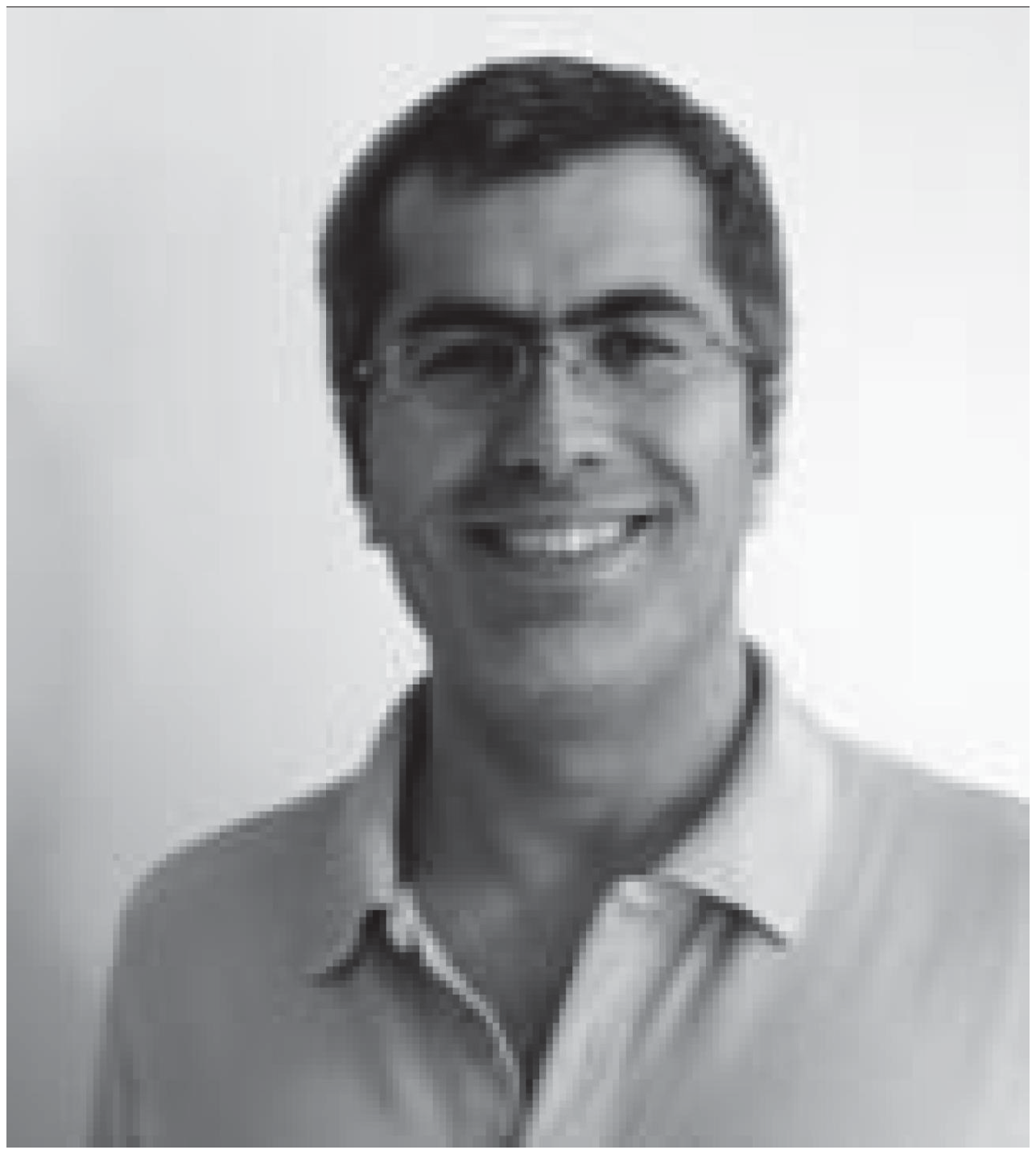}}]{Sahin Kaya Ozdemir}
received his B.S. degree and M.S. degree in Department of Electrical and Electronics Engineering,
Middle East Technical University (METU), Ankara, Turkey, and Ph.D. degree in Department of Electrical and Electronics Engineering, Graduate School of Electronic
Science and Technology, Shizuoka University, Hamamatsu, Japan, in 1992, 1995, and 2000, respectively.

From 2000 to 2003, he was a Post-doctoral researcher at Quantum Optics and Quantum Information Lab., Dept. of Photoscience, School of
Advanced Sciences, The Graduate University for Advanced Studies (Sokendai), Hayama,
Japan. Since 2004, he has been an Post-doctoral Research Associate at Micro/Nano-Photonics Lab., Department of Electrical and Systems Engineering,
Washington University in St. Louis. His research interests include micro/nano phtonics, quantum optics, and quantum information processing.
\end{biography}

\begin{biography}[{\includegraphics[width=1in,height=1.25in,clip,keepaspectratio]{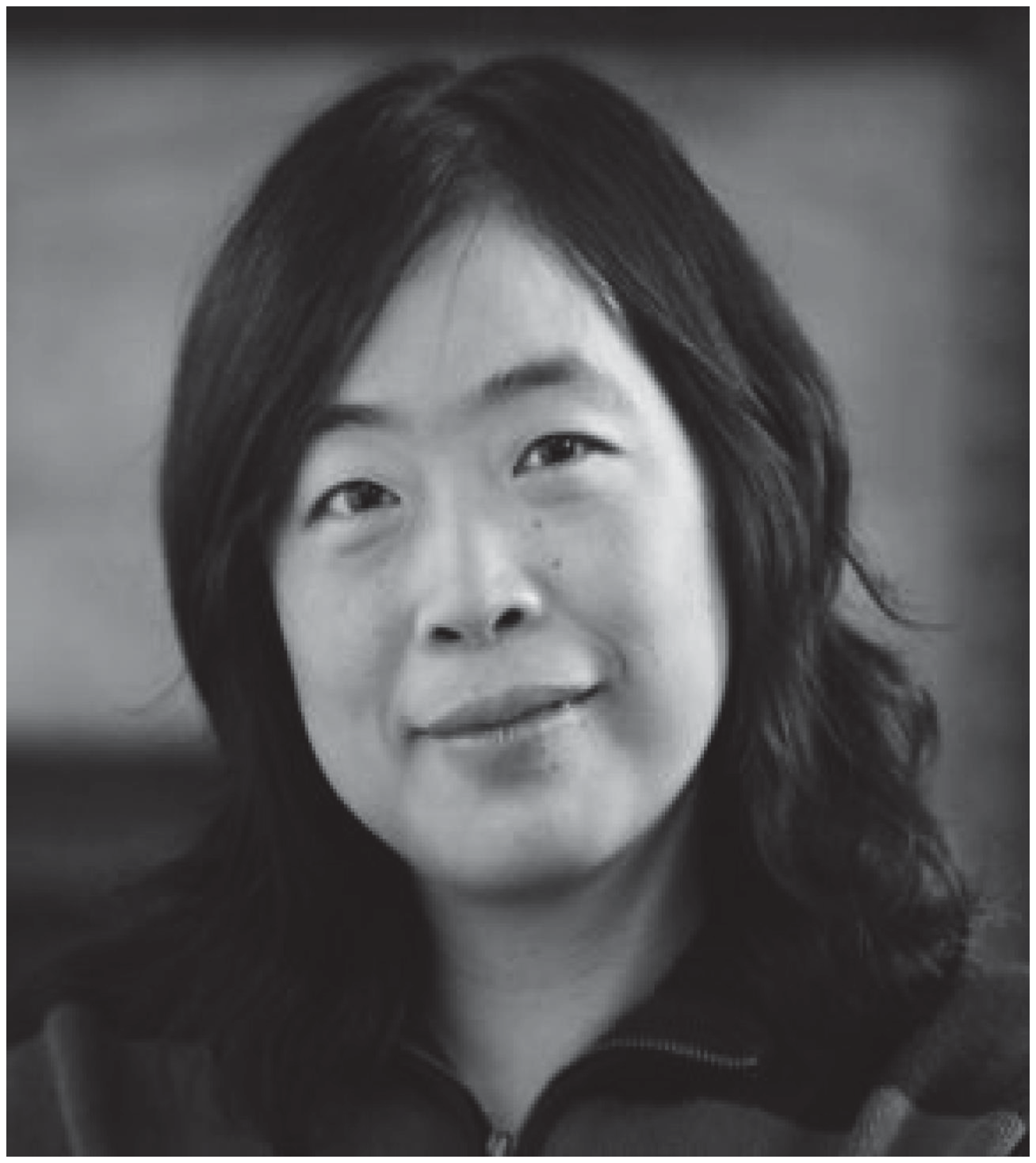}}]{Lan Yang}
received her B.S. degree in Materials Physics and M.S. degree in Solid State Physics from University of Science and Technology of China in 1997 and 1999,
and received her M.S. degree in Materials Science and Ph.D. degree in Applied Physics from Caltech in 2000 and 2005.

From 2005 to 2006, she was a Post-doctoral Scholar/Research Associate at Department of Applied Physics, Caltech. From
2007 to 2012, she was an Das Family Distinguished Career Development Assistant Professor at the Preston M. Green Department of Electrical and Systems Engineering, Washington University in St. Louis.
Since 2012, she has been an Associate Professor at the Preston M. Green Department of Electrical and Systems Engineering, Washington University in St. Louis.
She won the NSF CAREER Award and the Presidential Early Career Award for Scientists and Engineers (PECASE )in 2010. Her research interests include micro/nano phtonics and quantum optics.
\end{biography}

\begin{biography}[{\includegraphics[width=1in,height=1.25in,clip,keepaspectratio]{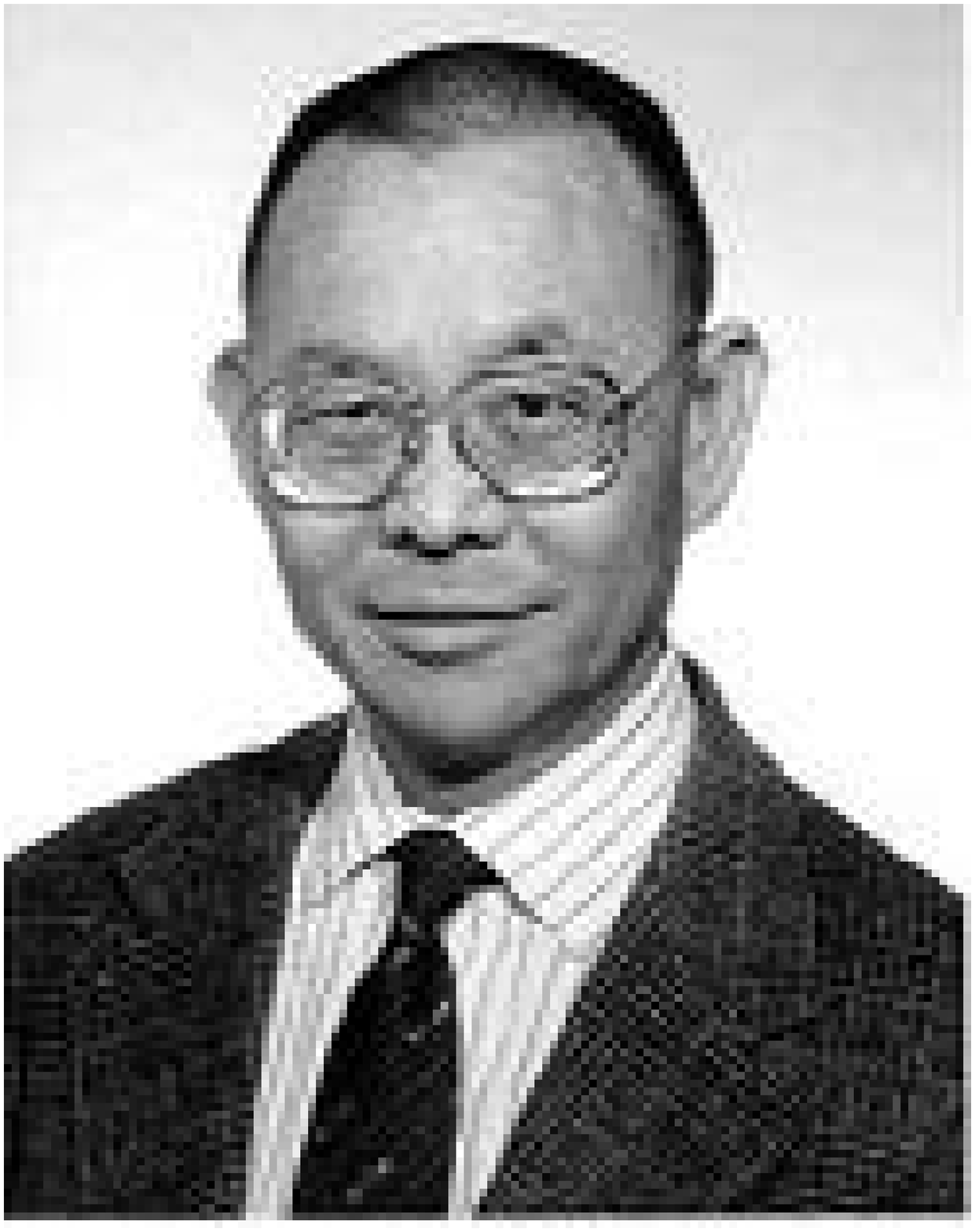}}]{Tzyh-Jong Tarn}
(M¡¯71-SM¡¯83-F¡¯85) received the D.Sc degree in control system
engineering from Washington University at St. Louis, Missouri, USA.

He is currently a Senior Professor in the Department of Electrical
and Systems Engineering at Washington University, St. Louis, USA.
He also is the director of the Center for Quantum Information
Science and Technology at Tsinghua University, Beijing, China.

An active member of the IEEE Robotics and Automation Society, Dr.
Tarn served as the President of the IEEE Robotics and Automation
Society, 1992-1993, the Director of the IEEE Division X (Systems
and Control), 1995-1996, and a member of the IEEE Board of
Directors, 1995-1996.

He is the first recipient of the Nakamura Prize (in recognition
and appreciation of his contribution to the advancement of the
technology on intelligent robots and systems over a decade) at the
10th Anniversary of IROS in Grenoble, France, 1997, the recipient
of the prestigious Joseph F. Engelberger Award of the Robotic
Industries Association in 1999 for contributing to the advancement
of the science of robotics, the Auto Soft Lifetime Achievement
Award in 2000 in recognition of his pioneering and outstanding
contributions to the fields of Robotics and Automation, the
Pioneer in Robotics and Automation Award in 2003 from the IEEE
Robotics and Automation Society for his technical contribution in
developing and implementing nonlinear feedback control concepts
for robotics and automation, and the George Saridis Leadership
Award from the IEEE Robotics and Automation Society in 2009. In
2010 he received the Einstein Chair Professorship Award from the
Chinese Academy of Sciences and the John R. Ragazzini Award from
the American Automatic Control Council. He was featured in the
Special Report on Engineering of the 1998 Best Graduate School
issue of US News and World Report and his research accomplishments
were reported in the ¡°Washington Times¡±, Washington D.C., the
¡°Financial Times¡±, London, ¡°Le Monde¡±, Paris, and the
¡°Chicago Sun-Times¡±, Chicago, etc. Dr. Tarn is an IFAC Fellow.

\end{biography}

\begin{biography}[{\includegraphics[width=1in,height=1.25in,clip,keepaspectratio]{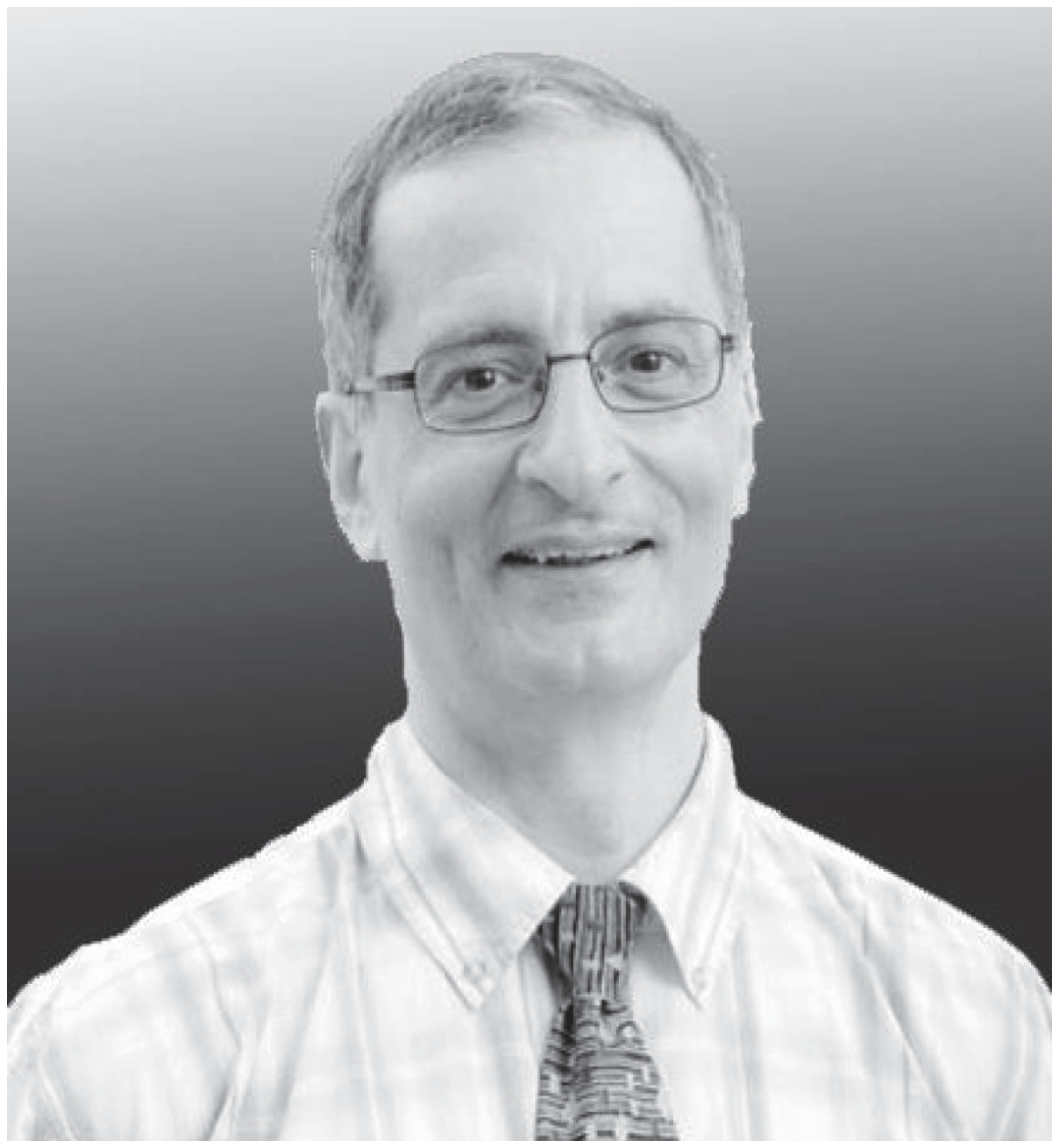}}]{Franco Nori}
received his M.S. and Ph.D. in Physics from the University of Illinois at Urbana-Champaign, USA, in 1982 and 1987.
From 1987 to 1989, he was a Postdoctoral Research Fellow at the Institute for Theoretical Physics,
University of California, Santa Barbara. Since 1990, he has been Assistant Professor, Associate
Professor, Full Professor and Research Scientist at the Department of Physics, University of Michigan,
Ann Arbor, USA. Also, since 2002, he has been a Team Leader at the Advanced Science Institute,
RIKEN, Saitama, Japan.  Since 2013, he is a RIKEN Chief Scientist, as well as a
Group Director of the Quantum Condensed Matter Research Group, at CEMS, RIKEN.

In 1997 and 1998, he received the ``Excellence in Education Award" and ``Excellence in Research Award" from the Univ. of Michigan.
In 2002, he was elected Fellow of the American Physical Society (APS), USA.
In 2003, he was elected Fellow of the Institute of Physics (IoP), UK.
In 2007, he was elected Fellow of the American Association for the
Advancement of Science (AAAS), USA. In 2013 he received the
Prize for Science and Technology, the Commendation for Science
and Technology, by the Minister of Education, Culture, Sports,
Science and Technology, Japan.

His research interests include nano-science, condensed matter physics,
quantum circuitry, quantum information processing, the dynamics of
complex systems, and the interface between mesoscopics,
quantum optics, atomic physics, and nano-science.

\end{biography}







\end{document}